\documentclass[11pt]{article}

\usepackage{color}
\usepackage{algorithmic,algorithm}
\usepackage{bm}
\usepackage{xspace}
\usepackage{subfigure}
\usepackage{color-edits}
\usepackage{enumitem}
\usepackage{amsmath}
\usepackage{graphicx}
\usepackage{fullpage}

\addauthor{xh}{blue}
\addauthor{dk}{red}

\newcommand{\Omit}[1]{}

\newcommand{\InfMax}{{Influence Maximization}\xspace}

\newcommand{\IntPer}{{Perturbation Interval}\xspace}
\newcommand{\RobInfMax}{{Robust Influence Maximization}\xspace}

\newcommand{\MINTSS}{\textsc{Greedy Mintss}\xspace}
\newcommand{\SATMINTSS}{\textsc{Saturate Greedy}\xspace}
\newcommand{\CONTINEST}{\textsc{ConTinEst}\xspace}
\newcommand{\GREEDY}{\textsc{Single Greedy}\xspace}
\newcommand{\ALLGREEDY}{\textsc{All Greedy}\xspace}
\newcommand{\SETCOVER}{\textsc{Set Cover}\xspace}
\newcommand{\MULTITREE}{\textsc{MultiTree}\xspace}

\newcommand{\ActProb}{\ensuremath{p}\xspace}
\newcommand{\ActProbD}[1]{\ensuremath{\ActProb_{#1}}\xspace}

\newcommand{\DelayP}[1]{\ensuremath{\alpha_{#1}}\xspace}
\newcommand{\Delay}[1]{\ensuremath{\Delta_{#1}}\xspace}

\newcommand{\IMFuncSym}[1][]{\ensuremath{\ifthenelse{\equal{#1}{}}{\sigma}{\sigma_{#1}}}\xspace}
\newcommand{\IMFUNC}[1][]{\ensuremath{\IMFuncSym[#1]}\xspace}
\newcommand{\IMFunc}[1]{\ensuremath{\IMFuncSym(#1)}\xspace}

\newcommand{\IMFuncSymTrue}{\ensuremath{\hat{\sigma}}\xspace}
\newcommand{\IMTRUE}{\ensuremath{\IMFuncSymTrue}\xspace}
\newcommand{\STRUE}{\ensuremath{\hat{S}}\xspace}
\newcommand{\OptS}[1]{\ensuremath{S^*_{#1}}\xspace}
\newcommand{\GreS}[1]{\ensuremath{S^g_{#1}}\xspace}
\newcommand{\IMFuncTrue}[1]{\ensuremath{\IMFuncSymTrue(#1)}\xspace}

\newcommand{\IMFSET}{\ensuremath{\Sigma}\xspace} % set of functions

\newcommand{\WCOBJ}{\ensuremath{\rho}\xspace} % robust objective
\newcommand{\WCObj}[1]{\ensuremath{\WCOBJ(#1)}\xspace} % with argument
\newcommand{\WCOBJG}{\ensuremath{\WCOBJ^g}\xspace} % greedy version
\newcommand{\WCObjG}[1]{\ensuremath{\WCOBJG(#1)}\xspace} % with argument

\newcommand{\I}{\ensuremath{I}\xspace}
\newcommand{\ID}[1]{\ensuremath{\I_{#1}}\xspace}
\newcommand{\UBD}[1]{\ensuremath{r_{#1}}\xspace}
\newcommand{\LBD}[1]{\ensuremath{\ell_{#1}}\xspace}
\newcommand{\SeedS}{\ensuremath{S_0}\xspace}
\DeclareMathOperator{\argmax}{argmax}
\providecommand{\Kth}[1]{\ensuremath{{#1}^{\rm th}}}
\providecommand{\Norm}[2][]{\ensuremath{%
\ifthenelse{\equal{#1}{}}{\|{#2}\|}{\|{#2}\|_{{#1}}}}\xspace}

\providecommand{\Set}[2][]{\ensuremath{%
\ifthenelse{\equal{#1}{}}{\{#2\}}{\{#2 \; | \; #1\}}}\xspace}

\providecommand{\Expect}[2][]{\ensuremath{%
\ifthenelse{\equal{#1}{}}{\mathbb{E}}{\mathbb{E}_{#1}}%
\left[#2\right]}\xspace}

\providecommand{\PROB}{\ensuremath{{\rm Prob}}\xspace}
\providecommand{\Prob}[2][]{\ensuremath{%
\ifthenelse{\equal{#1}{}}{\PROB[#2]}{\PROB_{#1}[#2]}}\xspace}

\newtheorem{theorem}{Theorem}
\newtheorem{lemma}[theorem]{Lemma}

\newtheorem{definition}{Definition}

\begin{document}

\begin{titlepage}
\title{Robust Influence Maximization}

\author{
Xinran He \qquad \qquad David Kempe\\[2ex]
\normalsize
Department of Computer Science\\
University of Southern California\\
Los Angeles, CA, United States\\
e-mail: \{xinranhe,dkempe\}@usc.edu
}

\maketitle
\begin{abstract}
Uncertainty about models and data is ubiquitous in the computational
social sciences,
%Models for processes on social networks are hard or impossible to
%test, and data about social networks must usually be inferred from
%imprecise observations or reports.
%Such uncertainty 
and it creates a need for \emph{robust} social network algorithms,
which can simultaneously provide guarantees across a spectrum of
models and parameter settings.
We begin an investigation into this broad domain by studying robust
algorithms for the \InfMax problem, in which the goal is
to identify a set of $k$ nodes in a social network whose joint
influence on the network is maximized.

We define a \RobInfMax framework wherein an algorithm is presented
with a set of influence functions, typically derived from different
influence models or different parameter settings for the same model. 
The different parameter settings could be derived from observed cascades
on different topics, under different conditions, or at different times.
The algorithm's goal is to identify a set of $k$ nodes who are
simultaneously influential for all influence functions, compared to the
(function-specific) optimum solutions.

We show strong approximation hardness results for this problem unless
the algorithm gets to select at least a logarithmic factor more seeds
than the optimum solution. 
However, when enough extra seeds may be selected, we show that
techniques of Krause et al.~can be used to approximate the optimum
robust influence to within a factor of $1-1/e$.
We evaluate this bicriteria approximation algorithm against natural
heuristics on several real-world data sets.
Our experiments indicate that the worst-case hardness does not
necessarily translate into bad performance on real-world data sets;
all algorithms perform fairly well. 
%We also highlight qualitative differences between algorithms
%deliberately optimizing for robustness and those for which 
%robustness is an afterthought.

\end{abstract}
\end{titlepage}

\section{Introduction}
Computational social science is the study of social and economic
phenomena based on electronic data, algorithmic approaches and
computational models. 
It has emerged as an important application of data mining and learning,
while also invigorating research in the social sciences.
%Beyond the intrinsic interest held by understanding such phenomena,
Computational social science is frequently envisioned as a
foundation for a discipline one could term ``computational social
engineering,'' wherein algorithmic approaches are used to
change or mitigate individuals' behavior.
%in a way that may lead to societally or financially desirable outcomes.

Among the many concrete problems that have been studied in this context,
perhaps the most popular is \InfMax. 
It is based on the observation that behavioral change in individuals is
frequently effected by influence from their social contacts.
Thus, by identifying a small set of ``seed nodes,''
one may influence a large fraction of the social network.
The desired behavior may be of social value, such as refraining
from smoking or drug use, using superior crops, or following hygienic
practices.
Alternatively, the behavior may provide financial value, as in the
case of viral marketing, where a company wants to rely on word-of-mouth
recommendations to increase the sale of its products.

\subsection{Prevalence of Uncertainty and Noise}
Contrary to the ``hard'' sciences, the study of social networks \xhedit{---}
whether using traditional or computational approaches \xhedit{---}
suffers from massive amounts of noise inherent in the data and models. 
The reasons range from the fundamental to the practical:
\begin{itemize}
\item At a fundamental level, it is not even clear what a
``social tie'' is. 
Different individuals or researchers operationalize 
the intuition behind ``friendship'', ``acquaintance'', ``regular''
advice seeking, etc.~in different ways (see, e.g., \cite{Campbell:Lee}).
Based on different definitions, the same real-world individuals and
behavior may give rise to different mathematical models of the same
``social network.''
\item Mathematical models of processes on social networks (such as
opinion adoption or tie formation) are at best approximations of
reality, and frequently mere guesses or mathematically convenient
inventions. Furthermore, the models are rarely validated against
real-world data, in large part due to some of the following concerns.
\item Human behavior is typically influenced by many environmental
variables, many of them hard or impossible to measure. 
Even with the rapid growth of available social data, 
it is unlikely that data sets will become sufficiently rich to 
%be able to understand 
disentangle the dependence of human behavior on the myriad variables
that may shape it.
\item Observational data on social behavior is virtually always incomplete. 
For example, even if API restrictions and privacy were not concerns
(which they definitely are at this time) and a ``complete'' data set
of Twitter \emph{and} Facebook \emph{and} e-mail communication were
collected, it would still lack in-person and phone interactions.
\item Inferring model parameters relies on a choice of model and
hyperparameters, many of which are difficult to make.
Furthermore, while for many models, parameter inference is
computationally efficient, this is not universally the case.
\end{itemize}

%\subsection{Dealing with Uncertainty}
Since none of these issues are likely to be resolved anytime
soon, both the models for social network processes and their inferred
parameters must be treated with caution.
This is true both when one wants to draw scientific insight for its
own sake, and when one wants to use the inferred models to make
computational social engineering decisions.
Indeed, the correctness guarantees for algorithms are predicated on 
the assumption of correctness of the model and the inferred parameters.
When this assumption fails --- which is inevitable --- the utility of
the algorithms' output is compromised.
Thus, to make good on the claims of real-world relevance of
computational social science,
\emph{it is imperative that the research community
focus on robustness as a primary design goal.}

\subsection{Modeling Uncertainty in Influence Maximization}
We take an early step in this bigger agenda, 
studying robustness in the context of the well-known \InfMax problem. 
(Detailed definitions are given in Section~\ref{sec:models}.) 
In \InfMax, the algorithm selects a set \SeedS of \emph{seed nodes},
of pre-specified size $k$.
The seed nodes are initially exposed to a product or idea; 
we say that they are \emph{active}. 
Based on a probabilistic model of influence propagation\footnote{We
  use the terms ``influence propagation'' and ``diffusion''
  interchangeably.}, 
they cause some of their neighbors to become active, who then cause
some of their neighbors to become active, etc.;
this process leads to a (random) final set of active nodes.
The goal is to maximize the size
%\footnote{Most results in the area
%  readily extend to assigning different nodes different non-negative
%  weights.} 
of this set; we denote this quantity by \IMFunc{\SeedS}.

The concerns discussed above combine to lead to significant
uncertainty about the function \IMFUNC:
different models give rise to very different functional forms of \IMFUNC, 
and missing observations or approximations in inference lead to
uncertainty about the models' parameters.

To model this uncertainty, we assume that the algorithm is presented
with a set \IMFSET of influence functions, 
and assured that one of these functions actually describes the
influence process, but not told \emph{which} one.
The set \IMFSET could be finite or infinite.
A finite \IMFSET could result from a finite set of different
information diffusion models that are being considered,  
or from of a finite number of different contexts under which
the individuals were observed (e.g., word-of-mouth cascades for
different topics or products), 
or from a finite number of different inference algorithms or algorithm
settings being used to infer the model parameters from observations.
An infinite (even continuous) \IMFSET arises if each model parameter
is only known to lie within some given interval;
this model of adversarial noise, which we call the \IntPer model, 
was recently proposed in \cite{InfluenceStability}.

Since the algorithm does not know \IMTRUE, in the \RobInfMax problem,
it must ``simultaneously optimize'' for all objective functions in
\IMFSET, in the sense of maximizing
%\begin{align} \label{eqn:robust-objective}
$\WCObj{\SeedS}  = \min_{\IMTRUE \in \IMFSET} 
\frac{\IMFuncTrue{\SeedS}}{\IMFuncTrue{\STRUE}}$,
%\end{align}
where $\STRUE \in \argmax_S \IMFuncTrue{S}$ is an optimal solution knowing
which function \IMTRUE is to be optimized. 
In other words, the selected set should simultaneously get as close as
possible to the optimal solutions for all possible objective functions.

\subsection{Our Approach and Results}
Our work is guided by the following overarching questions:
\begin{enumerate}
\item How well can the objective \WCOBJ be optimized in principle?
\item How well do simple heuristics perform in theory?
\item How well do simple heuristics perform in practice?
\item How do robustly and non-robustly optimized solutions differ
  qualitatively?
\end{enumerate}

We address these questions as follows. 
First, we show (in Section~\ref{sec:algorithm}) that unless the
algorithm gets to exceed the number of seeds $k$ by at least a factor
$\ln |\IMFSET|$, approximating the objective \WCOBJ to within a factor
$O(n^{1-\epsilon})$ is NP-hard for all $\epsilon > 0$.

However, when the algorithm does get to exceed the seed set target $k$
by a factor of $\ln |\IMFSET|$ (times a constant), much better 
bicriteria approximation guarantees can be obtained.\footnote{
A bicriteria algorithm gets to pick more nodes than the optimal
solution, but is only judged against the optimum solution with the
original bound $k$ on the number of nodes.}
Specifically, we show that a modification of an algorithm of
Krause et al.~\cite{krause:mcmahan:guestrin:gupta} 
%results in an algorithm that 
uses $O(k \ln |\IMFSET|)$ seeds and finds a seed set whose 
influence is within a factor $(1-1/e)$ of optimal.

%The algorithm has a running time that scales linearly in the
%size of the set \IMFSET of possible influence functions. 
%Thus, it cannot be used when \IMFSET is infinite. 
%We therefore also show how to reduce the model of parameter intervals
%to a finite set \IMFSET over which optimization is performed, by
%establishing that the worst-case input \IMTRUE always has each
%parameter as large or as small as possible, across a wide class of
%models. This approach reduces the running time of the Krause et
%al.~algorithm from infinite to exponential, and we present heuristics
%(without provable guarantees) to improve it much further in practice.

We also investigate two straightforward heuristics:
\begin{enumerate}
\item Run a greedy algorithm to optimize \WCOBJ
  directly, picking one node at a time.
\item For each objective function $\IMFUNC \in \IMFSET$, find a set
  $S_{\IMFUNC}$ (approximately) maximizing \IMFunc{S_{\IMFUNC}}.
  Evaluate each of these sets under \WCObj{S_{\IMFUNC}}, 
  and keep the best one.
\end{enumerate}

We first exhibit instances on which both of the heuristics perform very poorly.
Next (in Section~\ref{sec:experiment}), 
we focus on more realistic instances , exemplifying the types of scenarios
under which robust optimization becomes necessary.
In the first set of experiments, we infer influence networks on
a fixed node set from Twitter cascades on different \emph{topics}. 
Individuals' influence can vary significantly based on the topic, 
and for a previously unseen topic, 
it is not clear which inferred influence network to use.
In additional sets of experiments, we derive data sets from
  the same MemeTracker data~\cite{leskovec2009meme}, but use 
different \emph{time slices}, different \emph{inference algorithms and
parametrizations}, and different samples from confidence intervals.

The main outcome of the experiments is that while the algorithm with
robustness as a design goal typically (though not even always)
outperforms the heuristics, the margin is often quite small. 
Hence, heuristics may be viable in practice, when the influence
functions are reasonably similar.
A visual inspection of the nodes chosen by different algorithms
reveals how the robust algorithm ``hedges its bets'' across models,
while the non-robust heuristic tends to cluster selected nodes in one
part of the network.

\subsection{Stochastic vs.~Adversarial Models}
Given its prominent role in our model, the decision to treat the
choice of \IMTRUE as adversarial rather than stochastic deserves
some discussion.

First, adversarial guarantees are stronger than stochastic guarantees,
and will lead to more robust solutions in practice. 
Perhaps more importantly, inferring a Bayesian prior over influence
functions in \IMFSET will run into exactly the type of problem
we are trying to address in the first place: data are sparse and
noisy, and if we infer an incorrect prior, it may lead to very
suboptimal results. Doing so would next require us to establish
robustness over the values of the \emph{hyperparameters} of the
Bayesian prior over functions.

Specifically for the \IntPer model, one may be tempted to treat the
parameters as drawn according to some distribution over their possible range.
This approach was essentially  taken in
\cite{adiga:kuhlman:mortveit:vullikanti,goyal:bonchi:lakshmanan:dataInfMax}.
Adiga et al.~\cite{adiga:kuhlman:mortveit:vullikanti} assume
that for each edge $e$ independently, its presence/absence was
misobserved with probability $\epsilon$, 
whereas Goyal et al.~\cite{goyal:bonchi:lakshmanan:dataInfMax} assume
that for each edge, the actual parameter is perturbed with independent
noise drawn uniformly from a known interval.
In both cases, under the Independent Cascade model (for example), the
edge activation probability can be replaced with the \emph{expected}
edge activation probability under the random noise model, which will
provably lead to the exact same influence function \IMFUNC.
Thus, independent noise for edge parameters, drawn from a known
distribution, does not augment the model in the sense of capturing
robustness. In particular, it does not capture uncertainty
in a meaningful way.

To model the type of issues one would expect to arise in
real-world settings, at the very least, noise must be correlated
between edges. For instance, certain subpopulations may be inherently
harder to observe or have sparser data to learn from.
However, correlated random noise would result in a more complex
description of the noise model, and thus make it harder to actually
learn and verify the noise model. In particular, as discussed above,
this would apply given that the noise model itself must be learned
from noisy data.
\section{Related Work}
\label{sec:related}
Based on the early work of Domingos and Richardson
\cite{domingos:richardson,richardson:domingos},
Kempe et al.~\cite{InfluenceSpread} formally defined the
problem of finding a set of influential individuals as a discrete
optimization problem, proposing a greedy algorithm with a $1-1/e$
approximation guarantee for the Independent Cascade
\cite{goldenberg:libai:muller:talk,goldenberg:libai:muller:complex}
and Linear Threshold \cite{granovetter:threshold-models} models. 
A long sequence of subsequent work focused on more efficient
algorithms for \InfMax (both with and without approximation guarantees) 
and on broadening the class of models for which guarantees
can be obtained \cite{borgs:brautbar:chayes:lucier,chen:wang:yang:efficient,chen:yuan:zhang:scalable,InfluenceSpread,khanna:lucier:influence-maximization,mossel:roch:submodular,chen:wang:wang:prevalent,wang:cong:song:xie}.
See the recent book by Chen et al.~\cite{chen:lakshmanan:castillo:influence-maximization-book}
and the survey in \cite{InfluenceSpread} for more detailed overviews.

As a precursor to maximizing influence, one needs to infer the
influence function \IMFUNC from observed data.
The most common approach is to estimate the parameters of a particular
diffusion model~\cite{Abrahao:Chierichetti:Kleinberg:Panconesi:tracecomplexity,
  daneshmand14netrate,
  gomez-rodriguez:balduzzi:schoelkopf:uncovering,
  gomez-rodriguez:leskovec:krause:inferring,
  Harikrishna:David:Yaron:learnfunc,
  netrapalli2012learning,
  saito:kimura:ohara:motoda:selecting}. 
Theoretical bounds on the required sample complexity for many
diffusion models have been established, including
\cite{  Abrahao:Chierichetti:Kleinberg:Panconesi:tracecomplexity,
  Harikrishna:David:Yaron:learnfunc,
  netrapalli2012learning} for the Discrete-Time Independent Cascade (DIC) model, 
\cite{daneshmand14netrate} for the Continuous-Time Independent
  Cascade (CIC) model, and
\cite{Harikrishna:David:Yaron:learnfunc} for the Linear Threshold
model. 
However, it remains difficult to decide which diffusion
models fit the observation best.
Moreover, the diffusion models only serve as a rough approximation to
the real-world diffusion process. 
In order to sidestep the issue of diffusion models, 
Du et al.~\cite{nan:Yingyu:Maria-Florina:le:inflearn}
recently proposed to directly learn the influence function \IMFUNC
from the observations, without assuming any particular
diffusion model.
%In their work, 
They only assume that the influence function is a
weighted average of coverage functions. 
While their approach provides polynomial sample complexity, 
they require a strong technical condition on finding an
accurate approximation to the reachability distribution.
Hence, their work remains orthogonal to the issue of \RobInfMax.

Several recent papers take first steps toward \InfMax under
uncertainty. 
Goyal, Bonchi and Lakshmanan~\cite{goyal:bonchi:lakshmanan:dataInfMax}
and Adiga et al.~\cite{adiga:kuhlman:mortveit:vullikanti} study 
\emph{random} (rather than adversarial) noise models, in which either
the edge activation probabilities \ActProbD{u,v} are perturbed with
random noise \cite{goyal:bonchi:lakshmanan:dataInfMax}, 
or the presence/absence of edges is flipped with a known probability
\cite{adiga:kuhlman:mortveit:vullikanti}. 
Neither of the models truly extends the underlying diffusion models,
as the uncertainty can simply be absorbed into the probabilistic
activation process.

Another approach to dealing with uncertainty 
is to carry out multiple influence campaigns, and to use the
observations to obtain better estimates of the model parameters.
Chen et al.~\cite{Chen:Wang:Yuan:Wang} model the problem as a
combinatorial multi-armed bandit problem and use the UCB1
algorithm with regret bounds. 
Lei et al.~\cite{Lei:Maniu:Mo:Cheng:Senellart} instead incorporate
beta distribution priors over the activation probabilities
into the DIC model. 
They propose several strategies to update the posterior distributions
and give heuristics for seed selection in each trial so as to balance
exploration and exploitation.
Our approach is complementary: even in an exploration-based
  setting, there will always be residual uncertainty, in particular
  when exploration budgets are limited.

%The special case of the
The adversarial \IntPer model was recently
proposed in work of the authors~\cite{InfluenceStability}. 
The focus in that work was not on robust optimization, but on
algorithms for detecting whether an instance was likely to suffer from
high instability of the optimal solution.
Optimization for multiple scenarios was also recently used in work by
Chen et al.~on tracking influential nodes as the structure of the
graph evolves over time~\cite{ChenSHX15}.
However, the model explicitly allowed updating the seed set over time,
while our goal is simultaneous optimization.

Simultaneously to the present work, 
Chen et al.~\cite{chen:lin:tan:zhao:zhou} 
and Lowalekar et al.~\cite{Lowalekar:Varakantham:Kumar:RIM}
have been  studying the \RobInfMax problem 
under the \IntPer model \cite{InfluenceStability}.
Their exact formulations are somewhat different.
The main result of 
Chen et al.~\cite{chen:lin:tan:zhao:zhou}
is an analysis of the heuristic of choosing the best solution among
three candidates: 
make each edge's parameter as small as possible, as large as possible,
or equal to the middle of its interval. 
They prove \emph{solution-dependent} approximation
guarantees for this heuristic.  

The objective of Lowalekar et al.~\cite{Lowalekar:Varakantham:Kumar:RIM}
is to minimize the maximum regret instead of maximizing the minimum ratio. 
They propose a heuristic based on constraint generation ideas to solve
the robust influence maximization problem. 
The heuristic does not come with approximation guarantees; instead,
\cite{Lowalekar:Varakantham:Kumar:RIM} proposes a solution-dependent
measure of robustness of a given seed set. 
As part of their work, \cite{Lowalekar:Varakantham:Kumar:RIM} prove a
result similar to our Lemma~\ref{LEM:INTERVAL-ENDPOINTS}, showing that
the worst-case instances all have the largest or smallest possible
values for all parameters. 
\section{Models and Problem Definition}
\label{sec:models}

\subsection{Influence Diffusion Models}
For concreteness, we focus on two diffusion models:
the discrete-time Independent Cascade model (DIC) \cite{InfluenceSpread}
and the continuous-time Independent Cascade model (CIC) \cite{gomez-rodriguez:leskovec:krause:inferring}. 
Our framework applies to most other diffusion models; 
in particular, most of the concrete results carry over to the discrete and continuous
Linear Threshold models \cite{InfluenceSpread,saito:kimura:ohara:motoda:selecting}.

Under the DIC model, the diffusion process unfolds in discrete time
steps as follows: when a node $u$ becomes active in step $t$, it
attempts to activate all currently inactive neighbors in step $t+1$.
For each neighbor $v$, it succeeds with a known probability \ActProbD{u,v}; 
the \ActProbD{u,v} are the parameters of the model. 
If node $u$ succeeds, $v$ becomes active.
Once $u$ has made all its attempts, it does not get to make further
activation attempts at later times; of course, the node $v$ may well
be activated at time $t+1$ or later by some node other than $u$.

The CIC model describes a continuous-time process. 
Associated with each edge $(u,v)$ is a delay distribution with
parameter \DelayP{u,v}. 
When a node $u$ becomes newly active at time $t_u$,
for every neighbor $v$ that is still inactive, 
a delay time \Delay{u,v} is drawn from the delay distribution.
\Delay{u,v} is the duration it takes $u$ to activate $v$, which
could be infinite (if $u$ does not succeed in activating $v$).
Commonly assumed delay distributions include
the Exponential distribution or Rayleigh distribution. 
If multiple nodes $u_1, \ldots, u_\ell$ attempt to activate $v$, then
$v$ is activated at the earliest time $\min_i t_{u_i} + \Delay{u_i,v}$.
Nodes are considered activated by the process if they are activated
within a specified observation window $[0,T]$.

A specific instance is described by the class of its influence model 
(such as DIC, CIC, or others not discussed here in detail) 
and the setting
%\Params
of the model's parameters; 
in the DIC and CIC models above, the parameters would be
the influence probabilities \ActProbD{u,v} and the parameters
\DelayP{u,v} of the edge delay distributions, respectively.
Together, they completely specify the dynamic process; 
and thus a mapping \IMFUNC from initially active sets \SeedS to the expected
number\footnote{The model and virtually all results in the literature
  extend straightforwardly when the individual nodes are assigned
  non-negative importance scores.}
\IMFunc{\SeedS} of nodes active at the end of the process.
We can now formalize the \InfMax problem as follows:

\begin{definition}[\InfMax]
Maximize the objective \IMFunc{\SeedS}
subject to the constraint $|\SeedS| \leq k$.
\end{definition}

For most of the diffusion models studied in the literature, including
the DIC \cite{InfluenceSpread} and CIC \cite{DuSonGomZha13} models, 
it has been shown that \IMFunc{\SeedS} is a monotone and
submodular\footnote{Recall that a set function $f$ is monotone iff
  $f(S) \leq f(T)$ whenever $S \subseteq T$, and is submodular
  iff $f(S \cup \{x\}) - f(S) \geq f(T \cup \{x\}) - f(T)$ 
  whenever $S \subseteq T$.}
function of \SeedS. 
These properties imply that a greedy approximation algorithm
guarantees a $1-1/e$ approximation \cite{nemhauser:wolsey:fisher}.

\subsection{Robust Influence Maximization}
\label{sec:robust-definition}
The main motivation for our work is that often, \IMFUNC is not
precisely known to the algorithm trying to maximize influence.
There may be a (possibly infinite) number of candidate functions \IMFUNC, 
resulting from different diffusion models or parameter settings.
We denote the set of all candidate influence functions\footnote{For
  computation purposes, we assume that the functions are represented
  compactly, for instance, by the name of the diffusion model and all
  of its parameters.} by \IMFSET.
We now formally define the \emph{Robust Influence Maximization}
problem.

\begin{definition}[Robust Influence Maximization]
\label{def:robust}
Given a set \IMFSET of influence functions, maximize the objective
\begin{align*}
\WCObj{S} & = \min_{\IMFUNC \in \IMFSET} \frac{\IMFunc{S}}{\IMFunc{{\OptS{\IMFUNC}}}},
\end{align*}
subject to a cardinality constraint $|S| \leq k$. 
Here \OptS{\IMFUNC} is a seed set with $|\OptS{\IMFUNC}| \leq k$
maximizing \IMFunc{\OptS{\IMFUNC}}.
\end{definition}

A solution to the Robust Influence Maximization problem achieves a
large fraction of the maximum possible influence (compared to the
optimal seed set) under all diffusion settings simultaneously.
Alternatively, the solution can be interpreted as solving
the \InfMax problem when the function \IMFUNC is chosen from \IMFSET
by an adversary.

%While most of the positive results on \InfMax in the literature are
%derived for models under which the function \IMFUNC is submodular and monotone, 
While Definition~\ref{def:robust} per se does not require the 
$\IMFUNC \in \IMFSET$ to be submodular and monotone, these properties
are necessary to obtain positive results.
Hence, we will assume here that all $\IMFUNC \in \IMFSET$ are monotone
and submodular, as they are for standard diffusion models.
%However, in order to obtain any positive results for approximating of
%\WCOBJ, one will typically want to restrict the functions in such a way.
Notice that even then, \WCOBJ is the minimum of submodular functions,
and as such not necessarily submodular itself \cite{krause:mcmahan:guestrin:gupta}.

A particularly natural and important special case of
Definition~\ref{def:robust} is the \IntPer model recently proposed in
\cite{InfluenceStability}. 
Here, the influence model is known (for concreteness, DIC), but there
is uncertainty about its parameters. 
For each edge $e$, we have an interval $\ID{e} = [\LBD{e}, \UBD{e}]$,
and the algorithm only knows that the parameter (say, \ActProbD{e})
lies in \ID{e}; the exact value is chosen by an adversary.
Notice that \IMFSET is (uncountably) infinite under this model.
While this may seem worrisome, the following lemma shows that we only
need to consider finitely (though exponentially) many functions:

\begin{lemma} \label{LEM:INTERVAL-ENDPOINTS}
Under the \IntPer model for DIC\footnote{The result carries over
  with a nearly identical proof to the Linear Threshold model. We currently do not know if it also extends to the CIC model.}, 
the worst case for the ratio in \WCOBJ for
any seed set \SeedS is achieved by making each \ActProbD{e} equal to
\LBD{e} or \UBD{e}.
\end{lemma}

\begin{proof}
Fix one edge $\hat{e}$, and consider an assignment (fixed for now)
$\ActProbD{e} \in \ID{e}$ of activation probabilities to all edges
$e \neq \hat{e}$. 
Let $x \in \ID{\hat{e}}$ denote the (variable) activation probability for
edge $e$.
First, fix any seed set \SeedS, and define 
$f_{\SeedS}(x)$ to be the expected number of nodes activated by \SeedS
when the activation probabilities of all edges $e \neq \hat{e}$ are 
\ActProbD{e} and the activation probability of $\hat{e}$ is $x$.

We express $f_{\SeedS}(x)$ using the triggering set
\cite[Section 4.1]{InfluenceSpread} approach. 
Let $\mathcal{G}$ be the set of all possible directed graphs on the
given node set $V$. 
For any graph $G$, let $R_G(S)$ be the number of nodes reachable
from $S$ in $G$ via a directed path, 
and let $P(G)$ be the probability that graph $G$ is obtained when each
edge $e$ is present in $G$ independently with probability \ActProbD{e}
(or $x$, if $e = \hat{e}$).
By the triggering set technique \cite[Proof of Theorem 4.5]{InfluenceSpread}, 
we get that 
\begin{align*}
f_{\SeedS}(x) & = \sum_{G \in \mathcal{G}} P(G) \cdot R_G(S).
\end{align*}
The probabilities $P(G)$ for obtaining a graph $G$ are:
\begin{align*}
P(G) & = (1-x) \cdot \prod_{e \in G} \ActProbD{e} 
               \cdot \prod_{e \notin G, e \neq \hat{e}} (1-\ActProbD{e})
& \mbox{ when } \hat{e} \notin G;\\
P(G) & = x \cdot \prod_{e \in G, e \neq \hat{e}} \ActProbD{e} 
           \cdot \prod_{e \notin G} (1-\ActProbD{e})
& \mbox{ when } \hat{e} \in G.
\end{align*}
In either case, we obtain a linear function of $x$, so that
$f_{\SeedS}(x)$, being a sum of linear functions, is also linear in
$x$.

Therefore, the function $g(x) := \max_{\SeedS} f_{\SeedS}(x)$,
being a maximum of linear functions of $x$, is convex and piecewise linear.
Consider any fixed seed set \SeedS, and the ratio
$h(x) := \frac{f_{\SeedS}(x)}{g(x)}$.
Its $\alpha$-level set $\Set[h(x) \geq \alpha]{x}$ is equal to
$\Set[g(x) - 1/\alpha \cdot f_{\SeedS}(x) \leq 0]{x}$.
Because $g(x) - 1/\alpha \cdot f_{\SeedS}(x)$, a convex function minus
a linear function, is convex, its 0-level set is convex.
Hence, all $\alpha$-level sets of $h$ are convex, and $h$ is
quasi-concave.

Because $h$ is quasi-concave, it is unimodal, and thus minimized at
one of the endpoints of the interval.
Hence, we can minimize the ratio $h(x)$ --- and thus the performance
of the seed set \SeedS --- by making $x$ either as small or as large
as possible. 
By repeating this argument for all edges $\hat{e}$ one by one, we
arrive at an influence setting minimizing the performance of \SeedS,
and in which all influence probabilities are equal to the left or
right endpoint of the respective interval \ID{\hat{e}}. 
\end{proof}

%\begin{proof}
%Fix one edge $\hat{e}$, and a seed set \SeedS.
%Fix the activation probabilities on all edges except $\hat{e}$,
%and consider the function $f_{\SeedS}(x)$, the expected number of
%nodes activated by \SeedS when the activation probabilities of all
%edges $e \neq \hat{e}$ are as fixed, and the activation probability of
%$\hat{e}$ is $x$. 
%
%By explicitly writing the expectation as a distribution over live edge
%graphs (see \cite{InfluenceSpread}), one can observe that
%$f_{\SeedS}(x)$ is a linear function of $x$.
%Hence, the optimum influence (over all \SeedS) $g(x)$, being a maximum
%of linear functions, is convex.
%This means that the ratio $\frac{f_{\SeedS}(x)}{g(x)}$ is quasi-concave, which
%one can show by considering its level sets, and noticing that
%$\Set[f_{\SeedS}(x)/g(x) \geq \alpha]{x} = 
%\Set[g(x) - 1/\alpha \cdot f_{\SeedS}(x) \leq 0]{x}$.
%The latter is a level set of a convex function, and thus convex.
%
%Finally, because $f_{\SeedS}(x)/g(x)$ is quasi-concave, it attains its
%minimum at the smallest or largest value of $x$. 
%By repeating this argument for all edges $\hat{e}$, we obtain a
%worst-case setting in which all parameter values are equal to the left
%or right endpoints of the respective intervals \ID{e}. 
%\end{proof}

\section{Algorithms and Hardness} \label{sec:algorithm}
Even when \IMFSET contains just a single function \IMFUNC, 
\RobInfMax is exactly the traditional \InfMax problem, and is thus NP-hard.
This issue also appears in a more subtle way:
\emph{evaluating} \WCObj{\SeedS} (for a given \SeedS)
involves taking the minimum of
$\frac{\IMFunc{\SeedS}}{\IMFunc{\OptS{\IMFUNC}}}$ over all
$\IMFUNC \in \IMFSET$. 
It is not clear how to calculate the ratio
$\frac{\IMFunc{\SeedS}}{\IMFunc{\OptS{\IMFUNC}}}$ even for one of the
\IMFUNC, since the scaling constant \IMFunc{\OptS{\IMFUNC}} 
(which is independent of the chosen \SeedS) is exactly the solution to
the original \InfMax problem, and thus NP-hard to compute.

This problem, however, is fairly easy to overcome: 
instead of using the true optimum solutions \OptS{\IMFUNC} for the
scaling constants, we can compute $(1-1/e)$-approximations
\GreS{\IMFUNC} using the greedy algorithm, because the \IMFUNC are
monotone and submodular \cite{nemhauser:wolsey:fisher}.
Then, because $(1-1/e) \cdot \IMFunc{\OptS{\IMFUNC}}
\leq \IMFunc{\GreS{\IMFUNC}} \leq \IMFunc{\OptS{\IMFUNC}}$ for all
$\IMFUNC \in \IMFSET$, we obtain that the ``greedy objective
function''
\begin{align*}
\WCObjG{S} & = \min_{\IMFUNC \in \IMFSET} \frac{\IMFunc{S}}{\IMFunc{\GreS{\IMFUNC}}},
\end{align*}
satisfies the following property for all sets $S$:
\begin{align}
1-1/e) \cdot \WCObjG{S} & \leq \WCObj{S} \; \leq \; \WCObjG{S}.
\label{eqn:greedy-approximation}
\end{align}
Hence, optimizing \WCObjG{S} in place of \WCObj{S} comes at a cost of
only a factor $(1-1/e)$ in the approximation guarantee.
We will therefore focus on solving the problem of (approximately)
optimizing \WCObjG{S}.

Because each \IMFUNC is monotone and submodular, 
and the \IMFunc{\GreS{\IMFUNC}}, just like the \IMFunc{\OptS{\IMFUNC}},
are just scaling constants, 
\WCObjG{S} is a minimum of monotone submodular functions.
However, we show 
(in Theorem~\ref{THM:HARDNESS}, proved in Appendix~\ref{sec:hardness}) 
that even in the context of \InfMax, this minimum is
impossible to approximate to within any polynomial factor.
This holds even in a bicriteria sense, i.e., the
algorithm's solution is allowed to pick $(1-\delta) \ln |\IMFSET| \cdot k$
nodes, but is compared only to solutions using $k$ nodes.
The result also extends to the seemingly more restricted
\IntPer model, giving an almost equally strong bicriteria
approximation hardness result there.

\begin{theorem} \label{THM:HARDNESS}
Let $\delta, \epsilon > 0$ be any constants, 
and assume that $\mbox{P} \neq \mbox{NP}$.
There are no polynomial-time algorithms for the following problems:

\begin{enumerate}
\item Given $n$ nodes and a set \IMFSET of influence functions on these
nodes (derived from the DIC or CIC models), as well as a target size $k$.
Find a set $S$ of $|S| \leq (1 - \delta) \ln |\IMFSET| \cdot k$ nodes,
such that $\WCObj{S} \geq \WCObj{S^*} \cdot \Omega(1/n^{1-\epsilon})$,
where $S^*$ is the optimum solution of size $k$.
\item Given a graph $G$ on $n$ nodes and intervals \ID{e} for edge
activation probabilities under the DIC model (or intervals \ID{e} for edge delay parameters under the CIC model), as well as a target
size $k$.
Find a set $S$ of cardinality 
$|S| \leq \epsilon \cdot c \cdot \ln n \cdot k$ 
(for a sufficiently small fixed constant $c$)
such that $\WCObj{S} \geq \WCObj{S^*} \cdot \Omega(1/n^{1-\epsilon})$,
where $S^*$ is the optimum solution of size $k$.
\end{enumerate}
\end{theorem}

The hardness results naturally apply to any
diffusion model that subsumes the DIC or CIC models.
However, an extension to the DLT model is not immediate:
the construction relies crucially on having many edges of probability
1 into a single node, which is not allowed under the DLT model.

\subsection{Bicriteria Approximation Algorithm}
\label{sec:algorithm:bicriteria}
Theorem~\ref{THM:HARDNESS} implies that to obtain any non-trivial
approximation guarantee, one needs to allow the algorithm to exceed
the seed set size by at least a factor of $\ln |\IMFSET|$.
In this section, we therefore focus on such bicriteria approximation
results, by slightly modifying an algorithm of 
Krause et al.~\cite{krause:mcmahan:guestrin:gupta}.

The slight difference lies in how the submodular coverage
subproblem is solved. 
Both~\cite{krause:mcmahan:guestrin:gupta} and the \MINTSS
algorithm~\cite{goyal:bonchi:lakshmanan:venkatasubramanian:minimizing}
greedily add elements.
However, the \MINTSS algorithm adds elements until
the desired submodular objective is attained up to an additive
$\varepsilon$ term, while~\cite{krause:mcmahan:guestrin:gupta}
requires exact coverage. 
Moreover, directly considering real-valued submodular
functions instead of going through fractional values leads to a more 
direct analysis of the \MINTSS
algorithm~\cite{goyal:bonchi:lakshmanan:venkatasubramanian:minimizing}.

The high-level idea of the algorithm is as follows.
Fix a real value $c$, and define
$h^{(c)}_{\IMFUNC}(S) := \min(c, \frac{\IMFunc{S}}{\IMFunc{\GreS{\IMFUNC}}})$
and $H^{(c)}(S) := \sum_{\IMFUNC \in \IMFSET} h^{(c)}_{\IMFUNC}(S)$.
Then, $\WCObjG{S} \geq c$ if and only if
$\frac{\IMFunc{S}}{\IMFunc{\GreS{\IMFUNC}}} \geq c$ for 
all $\IMFUNC \in \IMFSET$.
But because by definition, $h^{(c)}_{\IMFUNC}(S) \leq c$ for all \IMFUNC,
the latter is equivalent to $H^{(c)}(S) \geq |\IMFSET| \cdot c$.
(If any term in the sum is less than $c$, no other term can ever
compensate for it, because they are capped at $c$.)

Because $H^{(c)}(S)$ is a non-negative linear combination of the
monotone submodular functions $h^{(c)}_{\IMFUNC}$, 
it is itself a monotone and submodular function.
This enables the use of a greedy $\ln |\IMFSET|$-approximation
algorithm to find an (approximately) smallest set $S$ with 
$H^{(c)}(S) \geq c |\IMFSET|$.
If $S$ has size at most $k \ln |\IMFSET|$, 
this constitutes a satisfactory solution, and we move on to larger values of $c$.
If $S$ has size more than $k \ln |\IMFSET|$, then the greedy
algorithm's approximation guarantee ensures that there is no
satisfactory set $S$ of size at most $k$. 
Hence, we move on to smaller values of $c$.
For efficiency, the search for the right value of $c$ is done with
binary search and a specified precision parameter.

A slight subtlety in the greedy algorithm is that $H^{(c)}$ could take
on fractional values. 
Thus, instead of trying to meet the bound $c |\IMFSET|$ precisely,
we aim for a value of $c |\IMFSET| - \epsilon$.
Then, the analysis of the \MINTSS algorithm of 
Goyal et al.~\cite{goyal:bonchi:lakshmanan:venkatasubramanian:minimizing}
(of which our algorithm is an unweighted special case) applies.
The resulting algorithm \SATMINTSS is given as Algorithm~\ref{Alg:RIF}. 
The simple greedy subroutine --- a special case of the \MINTSS
algorithm --- is given as Algorithm~\ref{Alg:GM}.

\begin{algorithm}[!t]
\caption{\SATMINTSS ($\IMFSET$, $k$, precision $\gamma$)\label{Alg:RIF}}
\begin{algorithmic}[1]
%\STATE \textbf{Inputs:} A set of diffusion settings $\mathcal{D}=\{D_1,\ldots,D_m\}$, number of seeds $k$
%\\\ \ \ \ \ \ \ \ \ \ \ \ precision $\gamma$
%\STATE \textbf{Outputs:} Selected seed set $S^*$
\STATE Initialize $c_{\min}\leftarrow 0$, $c_{\max}\leftarrow 1$.
\WHILE{$(c_{\max}-c_{\min}) \geq \gamma$}
\STATE $c \leftarrow (c_{\max}+c_{\min})/2$.
\STATE Define $H^{(c)}(S) \leftarrow 
\sum_{\IMFUNC \in \IMFSET} \min (c, \frac{\IMFunc{S}}{\IMFunc{\GreS{\IMFUNC}}})$.
\STATE $S \leftarrow \mbox{\MINTSS} (H^{(c)}, k, c \cdot |\IMFSET|, c \cdot \gamma/3)$.
\IF{$|S| > \beta \cdot k$}
\STATE $c_{\max} \leftarrow c$.
\ELSE
\STATE $c_{\min} \leftarrow c \cdot (1 - \gamma/3)$, 
$S^* \leftarrow S$.
\ENDIF
\ENDWHILE
\STATE Return $S^*$.
\end{algorithmic}
\end{algorithm}
\begin{algorithm}[!t]
\caption{\MINTSS($f$, $k$, threshold $\eta$, error $\varepsilon$) \label{Alg:GM}}
\begin{algorithmic}[1]
%\STATE \textbf{Inputs:} Monotone submodular function $f(s)$, number of elements $k$ \\\ \ \ \ \ \ \ \ \ \ \ \ coverage threshold $\eta$, error $\varepsilon$
%\STATE \textbf{Outputs:} Selected item set $S$
\STATE Initialize $S \leftarrow \emptyset$.
\WHILE{$f(S) < \eta-\varepsilon$}
\STATE $u \leftarrow \argmax_{v \notin S} f(S \cup \{v\})$.
%\STATE $u \leftarrow \argmax_{v} \min\{f(S \cup \{v\}), \eta\} - f(S)$.
%\STATE $u \leftarrow \argmax_{w \notin S} \min\{F(S\cup\{w\})\}$
\STATE $S \leftarrow S \cup \{u\}$.
\ENDWHILE
\STATE Return $S$.
\end{algorithmic}
\end{algorithm}

By combining the discussion at the beginning of this section (about
optimizing \WCOBJ vs.~\WCOBJG) with the analysis of 
Krause et al.~\cite{krause:mcmahan:guestrin:gupta} and 
Goyal et al.~\cite{goyal:bonchi:lakshmanan:venkatasubramanian:minimizing},
we obtain the following approximation guarantee.
%The (straightforward) missing details are given in
%Appendix~\ref{sec:satmintss-correctness}.

\begin{theorem} \label{THM:ALG}
Let $\beta = 1 + \ln |\IMFSET| + \ln \frac{3}{\gamma}$.
\SATMINTSS finds a seed set $\hat{S}$ of size $|\hat{S}| \leq \beta k$ with

\begin{align*}
\WCObj{\hat{S}} & \geq (1-1/e) \cdot \WCObj{S^*} - \gamma,
\end{align*}
where $S^* \in \argmax_{S:|S|\leq k} \WCObj{S}$ is an optimal robust seed
set of size $k$.
\end{theorem}
\begin{proof}
Algorithm~\ref{Alg:RIF} uses Algorithm~\ref{Alg:GM} (\MINTSS)
as a subroutine to find\footnote{%
Technically, the guarantees on \MINTSS depend on being able
  to evaluate $f$ precisely
  \cite[Theorem 1]{goyal:bonchi:lakshmanan:venkatasubramanian:minimizing}.
  However, 
  Theorem 2 of \cite{goyal:bonchi:lakshmanan:venkatasubramanian:minimizing}
  states that by obtaining $(1\pm \delta)$-approximations to $f$,
  we can ensure that 
  $|S| \leq (1+\delta') |S^*| \cdot (1+\ln\frac{\eta}{\varepsilon})$,
  where $\delta' \to 0$ as $\delta \to 0$.
  For influence coverage functions, arbitrarily close approximations
  to $f$ can be obtained by Monte Carlo simulations.
  We therefore ignore the issue of sampling accuracy in this article,
  and perform the analysis as though $f$ could be evaluated precisely.
  Otherwise, the approximations carry through in a straightforward
  way, leading to multiplicative factors $(1+\delta'')$.}
a solution $S$ such that 
$f(S) \geq \eta - \varepsilon$ and 
$|S| \leq |S^*| \cdot (1 + \ln \frac{\eta}{\varepsilon})$,
where $S^*$ is a smallest solution guaranteeing
$f(S^*) \geq \eta$.%

In light of the general outline and motivation for the \SATMINTSS
algorithm given above, it mostly
remains to verify how the guarantees for \MINTSS and the balancing of
the parameters carry through.

We will show that throughout the algorithm 
(or more precisely: the binary search), 
$c_{\min}$ always remains a lower bound on the solution for the
problem with the relaxed cardinality constraint,  
while $c_{\max}$ remains an upper bound on the solution for the
original problem.
In other words, there is no set $S$ of cardinality
at most $|S| \leq k$ with $\WCObj{S} > c_{\max}$,
and there \emph{is} a set $S$ of cardinality at most 
$|S| \leq \beta k$
with $\WCObj{S} \geq c_{\min}$.

To show this claim, consider the set $S$ returned by the \MINTSS
algorithm. 
If $|S| > \beta k$, the guarantee for \MINTSS implies that
$|S| \leq \beta |S^*|$, where $S^*$ is the optimal solution for
the instance.
Because $|S^*| \geq |S|/\beta > k$, the value $c$ is not feasible,
and the algorithm is correct in setting $c_{\max}$ to $c$.

Otherwise, $|S| \leq \beta k$, and the guarantee of \MINTSS implies
that $H^{(c)}(S) \geq c \cdot |\IMFSET| - c \cdot \gamma/3$.
Because each $h^{(c)}_{\IMFUNC}(S) \leq c$ by definition, we get
for all \IMFUNC,
\[
h^{(c)}_{\IMFUNC}(S) 
\; \geq \; H^{(c)}(S) - (|\IMFSET|-1) \cdot c
\; \geq \; c - c \cdot \gamma/3,
\]
and therefore $\WCObj{S} \geq c - c \cdot \gamma/3$.
This confirms the correctness of assigning 
$c_{\min}=c \cdot (1 - \gamma/3)$. 

Since we do not set $c_{\min} = c$, we need to briefly verify
termination of the binary search.
For any iteration in which we update $c_{\min}$,
let $\Delta := c_{\max} - c_{\min} \geq \gamma$.
When the new $c'_{\min}$ is set to 
$c \cdot (1-\gamma/3) \geq c - \gamma/3$, we get that 
$c_{\max} - c'_{\min} = (c_{\max} - c) + (c - c'_{\min})
\leq \Delta/2 + \gamma/3 \leq 5\Delta/6$.
Hence, the size of the interval keeps decreasing geometrically, and
the binary search terminates in $O(\log (1/\gamma))$ iterations.

At the time of termination, we obtain that
$|c^* - c_{\min}| \leq \gamma$.
Combining this bound with the factor of $(1-1/e)$ we lost due to
approximating \WCOBJ with \WCOBJG, we obtain the claim of the
theorem.
\end{proof}

Theorem~\ref{THM:ALG} holds very broadly, so long as all
influence functions are monotone and submodular. 
This includes the DIC, DLT, and CIC models, and allows mixing
  influence functions from different model classes.

Notice the contrast between Theorems~\ref{THM:ALG} and~\ref{THM:HARDNESS}.
By allowing the seed set size to be exceeded just a little more
(a factor $\ln |\IMFSET| + O(1)$ instead of $0.999 \ln |\IMFSET|$),
we go from $\Omega(n^{1-\epsilon})$ approximation hardness to
a $(1-1/e)$-approximation algorithm.

\subsection{Simple Heuristics}
\label{sec:simple-heuristics}
In addition to the \SATMINTSS algorithm, our experiments use two
natural baselines. The first is a simple greedy algorithm \GREEDY
which adds $\beta k$ elements to $S$ one by one, always choosing the one
maximizing $\WCObjG{S \cup \{v\}}$. 
While this heuristic has provable guarantees when the objective
function is submodular, this is not the case for the minimum of
submodular functions. 
%However, greedy algorithms are frequently chosen by practitioners and
%applied researchers, even when they are provably suboptimal.

The second heuristic is to run a greedy algorithm for each objective
function $\IMFUNC \in \IMFSET$ separately, and choose the best of the
resulting solutions. 
Those solutions are exactly the sets \GreS{\IMFUNC} defined earlier in
this section. 
Thus, the algorithm consists of choosing 
$\argmax_{\IMFUNC \in \IMFSET} \WCObjG{\GreS{\IMFUNC}}$.
We call the resulting algorithm \ALLGREEDY.

In the worst case, both \GREEDY and \ALLGREEDY can perform arbitrarily
badly, as seen by the following class of examples with a given
parameter $k$.
The example consists of $k$ instances of the DIC model for the
following graph with $3k+m$ nodes (where $m \gg k$).
The graph comprises a directed complete bipartite graph $K_{k,m}$
with $k$ nodes $x_1, \ldots, x_k$ on one side and $m$ nodes 
$y_1, \ldots, y_m$ on the other side, as well as
$k$ separate edges $(u_1, v_1), \ldots, (u_k, v_k)$. 
The edges $(u_i, v_i)$ have activation probability 1 in all instances.
In the bipartite graph, in the \Kth{i} scenario, only the edges
leaving node $x_i$ have probability 1, while all others have 0
activation probability.

The optimal solution for \RobInfMax is to select all nodes $x_i$,
since one of them will succeed in activating the $m$ nodes $y_j$.
The resulting objective value will be close to 1.
However, \ALLGREEDY only picks one node $x_i$ and the remaining
$k-1$ nodes as $u_j$. \GREEDY instead picks all of the $u_j$.
Thus, both \ALLGREEDY and \GREEDY will have robust influence close to
$0$ as $m$ grows large.
Empirical experiments confirm this analysis. 
For example, for $k=2$ and $m=100$, \SATMINTSS achieves
$\WCOBJ = 0.985$, while \GREEDY and \ALLGREEDY only achieve $0.038$
and $0.029$, respectively.

\subsubsection*{Implementation}
The most time-consuming step in all of the algorithms is the
estimation of influence coverage, given a seed set $S$.
Na\"{\i}ve estimation by Monte Carlo simulation could lead to
a very inefficient implementation. 
The problem is even more pronounced compared to traditional
\InfMax as we must estimate the influence in multiple diffusion settings. 
Instead, we use the \CONTINEST algorithm of 
Du et al.~\cite{DuSonGomZha13} for fast influence estimation under the CIC 
model. 
For the DIC model, we generalize the approach of Du et al.
To accelerate the \MINTSS algorithm, we also apply the CELF
optimization~\cite{LKGFVG} in all cases. 
Analytically, one can derive linear running time (in both $n$
  and $|\IMFSET|$) for all three algorithms, thanks to the fast influence estimation.
This is borne out by detailed experiments in Section~\ref{sec:scalability}.

\section{Experiments} \label{sec:experiment}
We empirically evaluate the \SATMINTSS algorithm and
the \GREEDY and \ALLGREEDY heuristics.
%\GREEDY represents a popular class of heuristics that does not fully
%recognize the properties of the objective function, 
%while \ALLGREEDY represents algorithms that do not consider the
%diffusion scenarios \emph{jointly}.
Our goal is twofold:
(1) Evaluate how well \SATMINTSS and the heuristics perform on realistic instances.
(2) Qualitatively understand the difference between robustly and
non-robustly optimized solutions.
%Our goal is threefold:
%(1) Understand how well \SATMINTSS performs in realistic instances.
%(2) Understand theoretically and practically how much worse simple
%heuristics perform.
%(3) Qualitatively understand the difference between robustly and
%non-robustly optimized solutions.

Our experiments are all performed on real-world data sets.
The exception is the scalability experiments 
in Section~\ref{sec:scalability}, which benefit from
the controlled environment of synthetic networks.
The data sets span the range of different causes for uncertainty, namely:
(1) influences are learned from cascades for different topics;
(2) influences are learned with different modeling assumptions;
(3) influences are only inferred to lie within intervals \ID{e} 
(the \IntPer model).

\subsection{Different Networks}
We first focus on the case in which the diffusion model is kept
constant: we use the DIC model, with parameters specified below.
Different objective functions are obtained from observing cascades
(1) on different topics. 
We use Twitter retweet networks for different topics.
(2) at different times. 
We use MemeTracker diffusion network snapshots at different times.

The Twitter networks are extracted from a complete collection of tweets
between Jan.~2010 and Feb.~2010.
We treat each hashtag as a separate cascade, and extract the top
100/250 users with the most tweets containing these hashtags into two
datasets (Twitter100 and Twitter250).
The hashtags are manually grouped into five categories of about 70--80
hashtags each, corresponding to major events/topics during the data
collection period.
The five groups are: Haiti earthquake (Haiti), Iran election (Iran),
Technology, US politics, and the Copenhagen climate change summit
(Climate). 
Examples of hashtags in each group are shown in Table~\ref{Tab:hashtags}. 
Whenever user $B$ retweets a post of user $A$ with a hashtag belonging
to category $i$, we insert an edge with activation probability 1 from
$A$ to $B$ in graph $i$. 
The union of all these edges specifies the \Kth{i} influence function.

Our decision to treat each hashtag as a separate cascade is
supposed to capture that most hashtags ``spread'' across Twitter
when one user sees another use it, and starts posting with it
himself. The grouping of similar hashtags captures that a user who may
influence another to use the hashtag, say, \#teaparty, would likely
also influence the other user to a similar extent to use, say,
\#liberty. The pruning of the data sets was necessary because most
users had showed very limited activity. Naturally, if our goal were to
evaluate the algorithmic efficiency rather than the performance with
respect to the objective function, we would focus on larger networks,
even if the networks were less easily visualized.

\begin{table}[htb]
\centering
\begin{tabular}{|l|l|}
    \hline
  Category & Hashtags\\   \hline \hline
  Iran & \#iranelection, \#iran, \#16azar, \#tehran \\ \hline
  Haiti & \#haiti, \#haitiquake, \#supphaiti, \#cchaiti \\ \hline
  Technology & \#iphone, \#mac, \#microsoft, \#tech \\ \hline
  US politics & \#obama, \#conservative, \#teaparty, \#liberty \\ \hline
  Climate & \#copenhagen, \#cop15, \#climatechange \\ \hline
\end{tabular}
\caption{Examples of hashtags in each category \label{Tab:hashtags}}
\end{table}

The MemeTracker dataset~\cite{leskovec2009meme} contains memes
extracted from the Blogsphere and main-stream media sites between
Aug.~2009 and Feb.~2010. 
In our experiments, we extract the 2000/5000 sites
with the most posting activity across the time period we study
(Meme2000 and Meme5000). 
We extract six separate diffusion networks, one for each month.
The network for month $i$ contains all the directed links that were
posted in month $i$ (in reverse order, i.e., if $B$ links to
  $A$, then we add a link from $A$ to $B$), with activation probability 1.
It thus defines the \Kth{i} influence function.

The parameters of the DIC model used for this set of experiments are
summarized in Table~\ref{Tab:setting1}. 

\begin{table}[h]
\centering
\begin{tabular}{|l|l|l|}
    \hline
  Data set & Edge Activation Probability & \# Seeds\\
  \hline
  \hline
  Twitter100 & 0.2 & 10\\ \hline
  Twitter250 & 0.1 & 20  \\ \hline
  Meme2000 & 0.05 & 50  \\ \hline
  Meme5000 & 0.05 & 100 \\ \hline
\end{tabular}
\caption{Diffusion model settings \label{Tab:setting1}}
\end{table}

Recalling that in the worst case, a relaxation in the number of seeds
is required to obtain robust seed sets, we allow all algorithms to
select more seeds than the solution they are compared against.
Specifically, we report results in which the algorithms may select 
$k$, $1.5 \cdot k$ and $2 \cdot k$ seeds, respectively.
The reported results are averaged over three independent runs of each
of the algorithms.

\subsubsection*{Results: Performance}
The aggregate performance of the different algorithms on the four data
sets is shown in Figure~\ref{Fig:performance-topics}.

\begin{figure*}[htb]
\setlength\fboxsep{0pt}
\begin{center}
\subfigure[Twitter100 ($k=10$)]{
\includegraphics[width=0.23\textwidth]{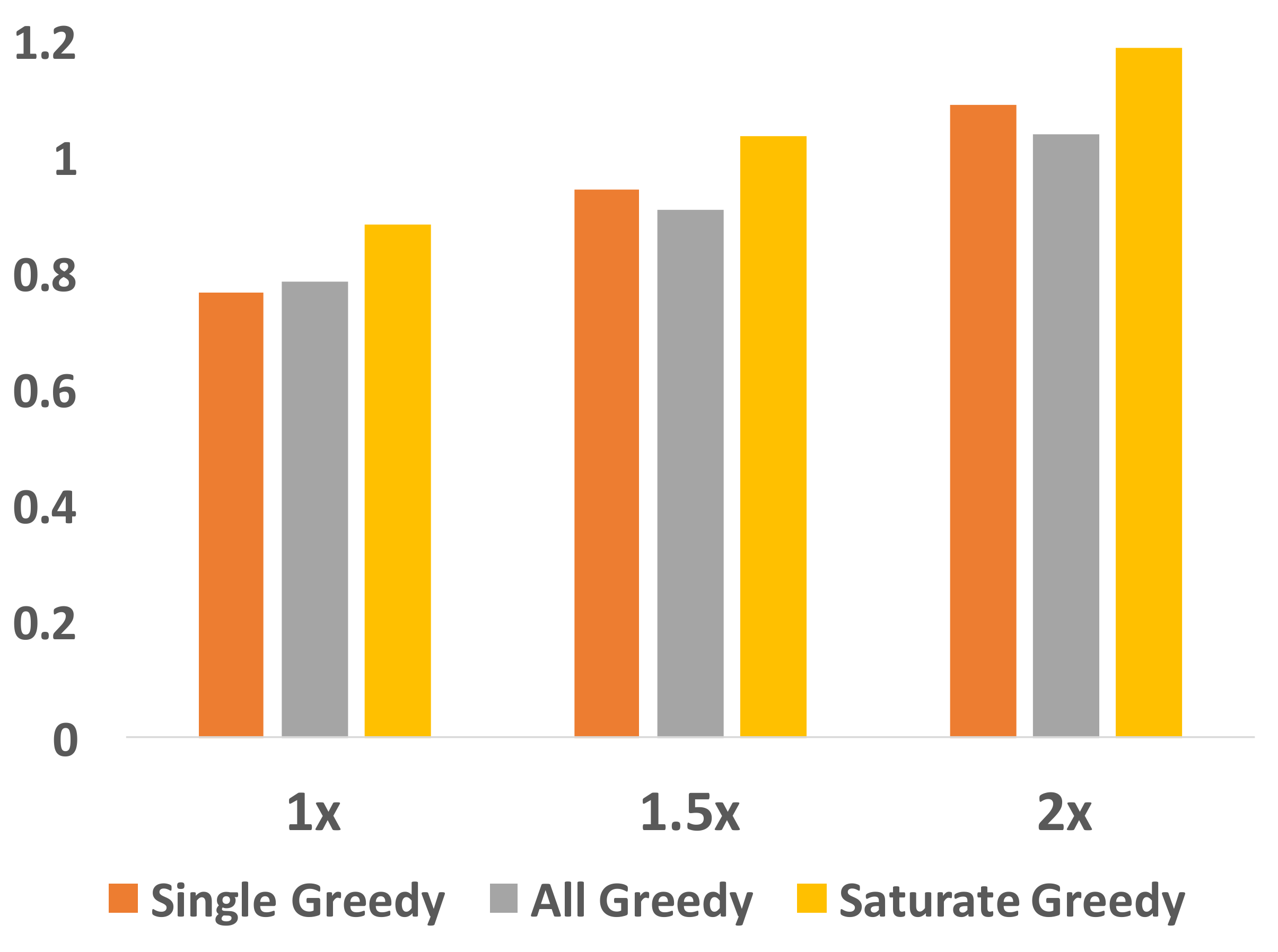}}
\subfigure[Twitter250 ($k=20$)]{
\includegraphics[width=0.23\textwidth]{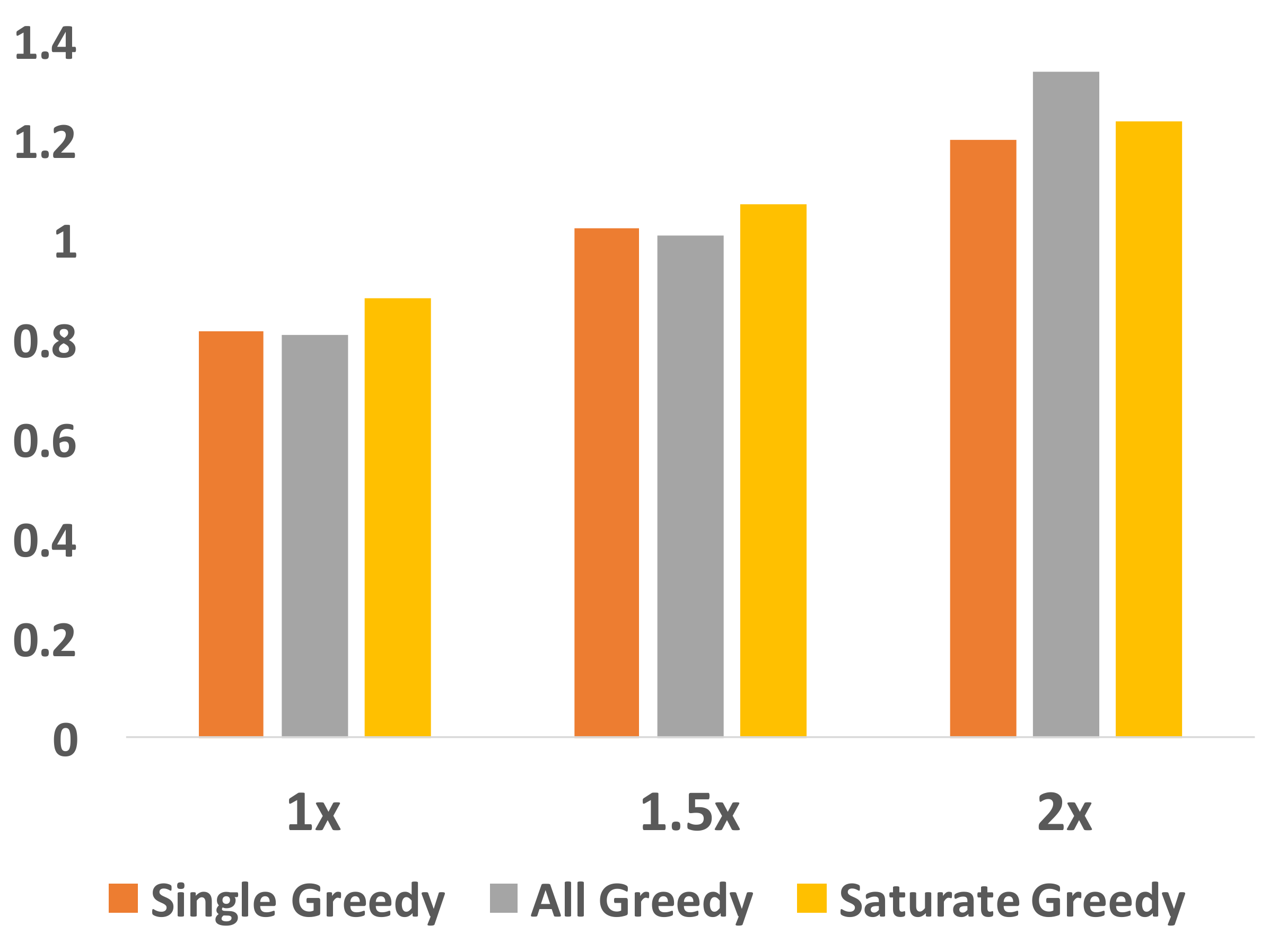}}
\subfigure[Meme2000 ($k=50$)]{
\includegraphics[width=0.23\textwidth]{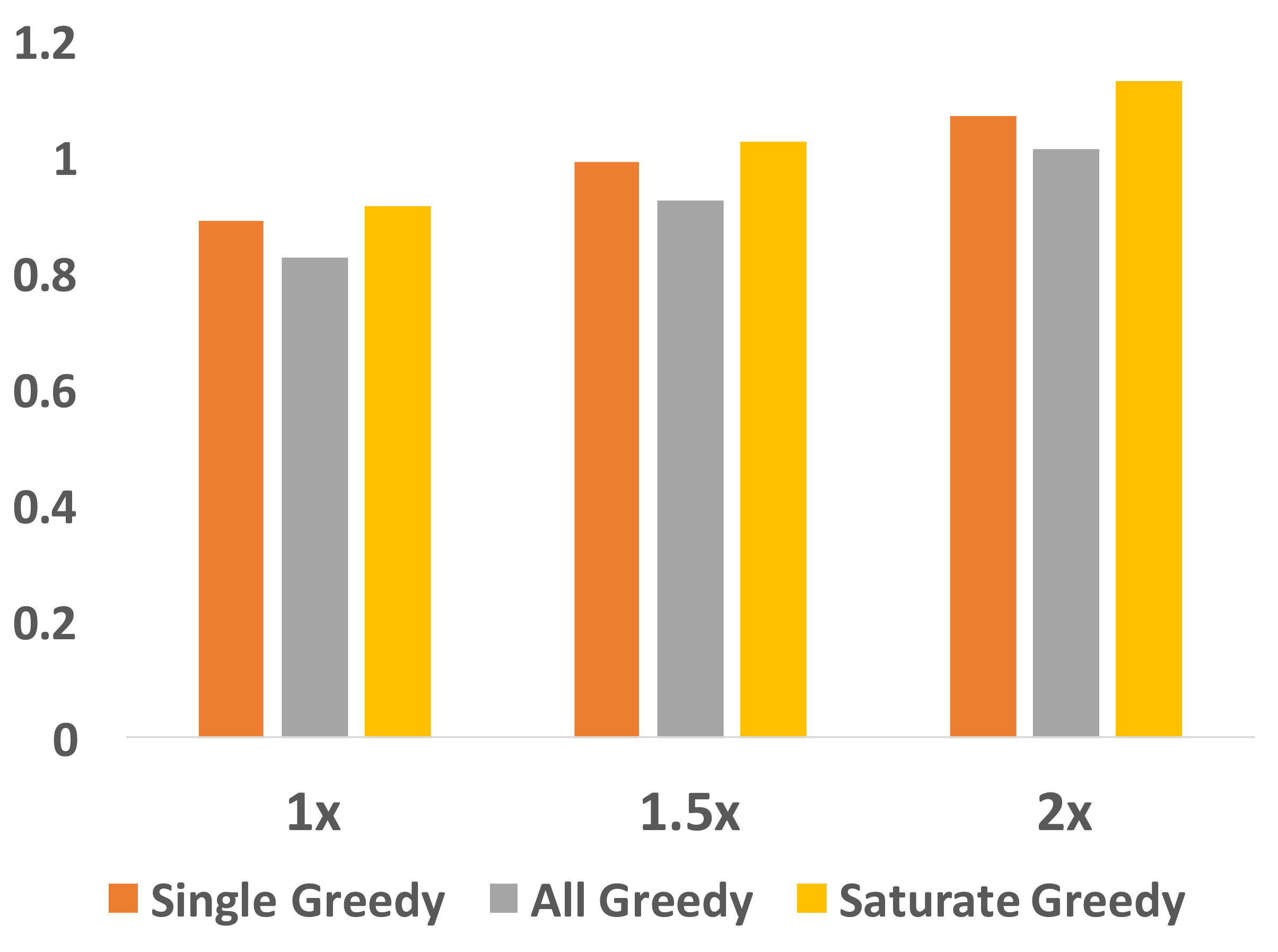}}
\subfigure[Meme5000 ($k=100$)]{
\includegraphics[width=0.23\textwidth]{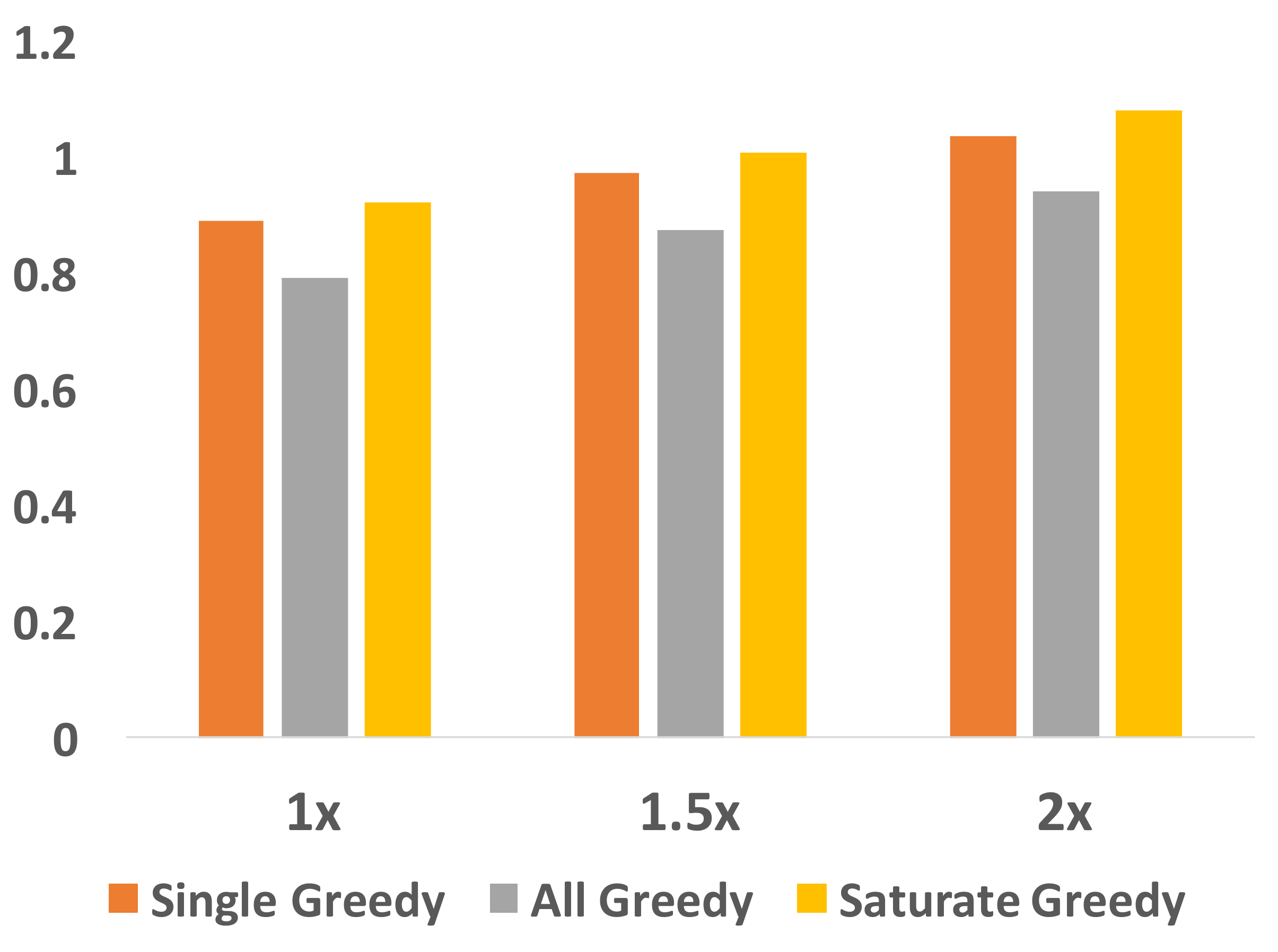}}
\end{center}
\caption{Performance of the algorithms on the four topical/temporal datasets. 
The $x$-axis is the number of seeds selected, and the $y$-axis
  the resulting robust influence (compared to seed sets of size
  $k$). \label{Fig:performance-topics}}
\end{figure*}

The first main insight is that (in the instances we study)
getting to over-select seeds by 50\%, all three algorithms
achieve a robust influence of at least 1.0. 
In other words, 50\% more seeds let the algorithms perform as though
they knew exactly which of the (adversarially chosen) diffusion
settings was the true one.
This suggests that the networks in our data sets share a lot of
similarities that make influential nodes in one network also (mostly)
influential in the other networks.
This interpretation is consistent with the observation that 
the baseline heuristics perform similarly to (and in one case better
than) the \SATMINTSS algorithm.
Notice, however, that when selecting just $k$ seeds, \SATMINTSS does
perform best (though only by a small margin) among the three
algorithms. 
This suggests that keeping robustness in mind may be more crucial
when the algorithm does not get to compensate with a larger number of
seeds.

\subsubsection*{Results: Visualization} 

To further illustrate the tradeoffs between robust and non-robust
optimization, we visualize the seeds selected by \SATMINTSS (robust
seeds) compared to seeds selected non-robustly based on only one
diffusion setting. 
For legibility, we focus only on the Twitter250 data set,
and only plot $4$ out of the $5$ networks. 
(The fifth network is very sparse, and thus not particularly interesting.)

Figure~\ref{Fig:vis_iran} compares the seeds selected by \SATMINTSS 
with those (approximately) maximizing the influence for the Iran
network. 
Notice that \SATMINTSS focuses mostly (though not exclusively)
on the densely connected core of the network (at the center), 
while the Iran-specific optimization also exploits the dense regions
on the left and at the bottom. These regions are much less densely connected
in the US politics and Climate networks, while the core remains fairly
densely connected, leading the \SATMINTSS solution to be somewhat more robust. 

\begin{figure*}[htb]
\setlength\fboxsep{0pt}
\begin{center}
\subfigure[Iran]{
\includegraphics[width=0.23\textwidth]{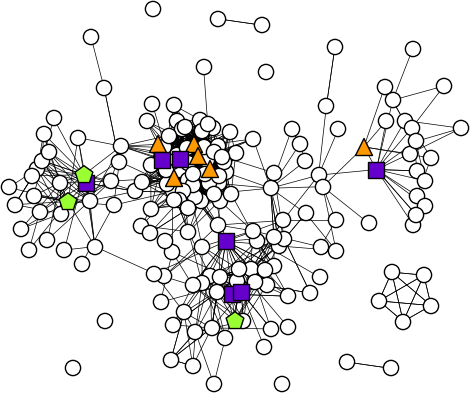}}
\subfigure[Haiti]{
\includegraphics[width=0.23\textwidth]{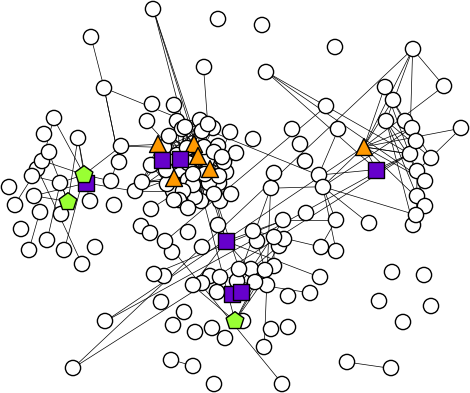}}
\subfigure[US politics]{
\includegraphics[width=0.23\textwidth]{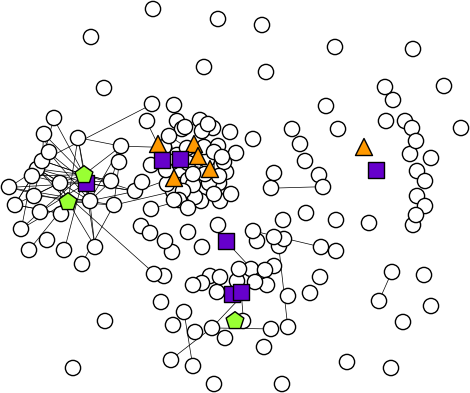}}
\subfigure[Climate]{
\includegraphics[width=0.23\textwidth]{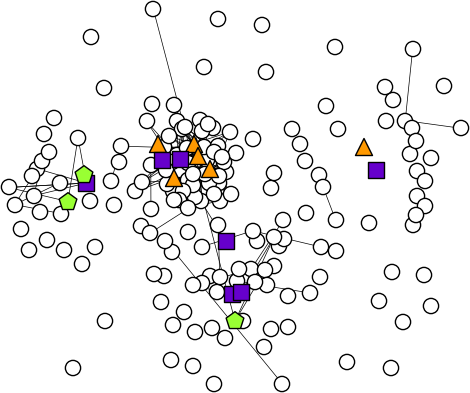}}
\end{center}
\caption{\SATMINTSS vs.~Iran graph seed nodes. Green/pentagon nodes are
  selected in both; orange/triangle nodes are selected by \SATMINTSS only; 
  purple/square nodes for Iran only.\label{Fig:vis_iran}}
\end{figure*}

Similarly, Figure~\ref{Fig:vis_climate} compares the \SATMINTSS
seeds (which are the same as in Figure~\ref{Fig:vis_iran}) 
with seeds for the Climate network. 
The trend here is exactly the opposite. 
The seeds selected based only on the Climate network are exclusively
in the core, because the other parts of the Climate network are
barely connected. On the other hand, the robust solution picks a few seeds
from the clusters at the bottom, left, and right, which are present
in other networks. These seeds lead to extra influence in those
networks, and thus more robustness. 

\begin{figure*}[htb]
\setlength\fboxsep{0pt}
\begin{center}
\subfigure[Iran]{
\includegraphics[width=0.23\textwidth]{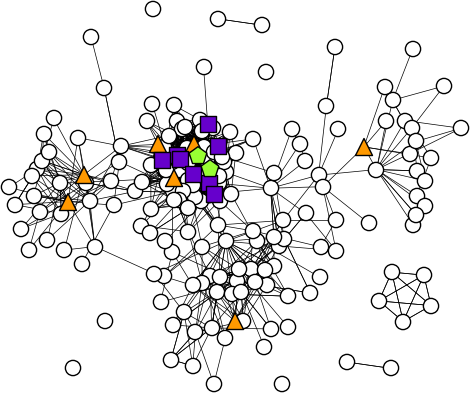}}
\subfigure[Haiti]{
\includegraphics[width=0.23\textwidth]{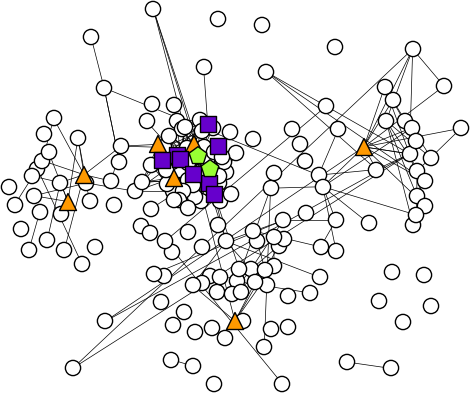}}
\subfigure[US politics]{
\includegraphics[width=0.23\textwidth]{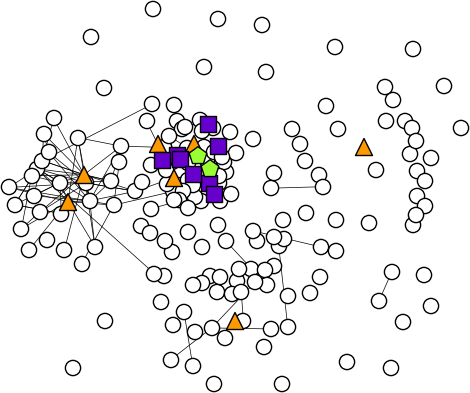}}
\subfigure[Climate]{
\includegraphics[width=0.23\textwidth]{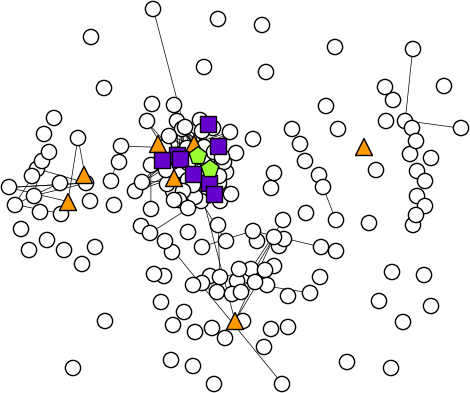}}
\end{center}
\caption{\SATMINTSS vs.~Climate graph seed nodes. Green/pentagon nodes are
  selected in both; orange/triangle nodes are selected by \SATMINTSS only; 
  purple/square nodes for Climate only.\label{Fig:vis_climate}}
\end{figure*}

\subsection{Different Diffusion Models}
%It is often difficult to decide which diffusion models to use for the
%\InfMax problem. 
In choosing a diffusion model, there is little convincing empirical
work guiding the choice of a model class (such as CIC, DIC, or
threshold models) or of distributional assumptions for model
parameters (such as edge delay).
A possible solution is to optimize robustly with respect to these
different possible choices. 

In this section, we evaluate such an approach.
Specifically, we perform two experiments:
(1) 
learning the CIC influence network under different parametric
assumptions about the delay distribution, and
(2) learning the influence network under different models of influence
(CIC, DIC, DLT).
We again use the MemeTracker dataset, restricting ourselves to the
data from August 2008 and the 500 most active users. 
We use the \MULTITREE algorithm of Gomez-Rodriguez et
al.~\cite{gomez-rodriguez:schoelkopf:multiple} to infer the diffusion
network from the observed cascades.
This algorithm requires a parametric assumption for the edge delay
distribution. 
We infer ten different networks $G_i$ corresponding to the 
Exponential distribution with parameters 0.05, 0.1, 0.2, 0.5, 1.0,
and to the Rayleigh distribution with parameters 0.5, 1, 2, 3, 4. 
The length of the observation window is set to 1.0.

We then use the three algorithms to perform robust influence
maximization for $k=10$ seeds, again allowing the algorithms to exceed
the target number of vertices. 
The influence model for each graph is the CIC model with the same
parameters that were used to infer the graphs.

The performance of the algorithms is shown in
Figure~\ref{Fig:diff_models}(a). 
All methods achieve satisfactory results in the experiment;
this is again due to high similarity between the different diffusion
settings inferred with different parameters. 

For the second experiment, we investigate
the robustness across different classes of diffusion models. 
We construct three instances of the DIC, DLT and DIC model from the 
ground truth diffusion network between the 500 active users. 
For the DIC model, we set the activation probability uniformly to $0.1$.
For the DLT model, we follow~\cite{InfluenceSpread} and set the
edge weights to $1/d_v$ where $d_v$ is the in-degree of node $v$. 
For the CIC model, we use an exponential distribution with parameter $0.1$ 
and an observation window of length $1.0$.
We perform robust influence maximization for $k=10$ seeds and again
allow the algorithms to exceed the target number of seeds.
 
The results are shown in Figure~\ref{Fig:diff_models}(b). 
Similarly to the case of different estimated parameters, 
all methods achieve satisfactory results in the experiment
due to the high similarity between the diffusion models.
Our results raise the intriguing question of which types of
networks would be prone to significant differences in algorithmic
performance based on which model is used for network estimation.
%\begin{figure}[htb]
%\center
%\includegraphics[width=0.4\textwidth]{figures/diff_model.pdf}
%\caption{Performance of \RobInfMax algorithms under different
%  diffusion models ($k=10$)\label{Fig:performance-model}}
%\end{figure}
\subsection{Networks sampled from the \IntPer model}
\label{sec:experiments:intper}
To investigate the performance when model parameters can only be
placed inside ``confidence intervals'' (i.e., the \IntPer model), we carry out
experiments under two networks, MemeTracker and STOCFOCS.

The MemeTracker network is extracted from the MemeTracker data set
using the \textsc{ConNIe} algorithm~\cite{myers:leskovec:convexity} 
to infer the (fractional) parameters (activation probabilities) of a
DIC model from the same
500-node MemeTracker data set used in the previous section.
We also ran experiments on a multigraph extracted from co-authorship
of published papers in the conferences STOC and FOCS from 1964--2001. 
Each node in that network is a researcher with at least one
publication in one of the conferences. 
For each multi-author paper, we add a complete undirected graph
among the authors. We compress parallel edges into a single edge with
weight $w_e$ and set the activation probability $\ActProbD{e} = 0.1\cdot w_e$. 
If $\ActProbD{e}>1$, we truncate its value to $1$.
%\footnote{%
%We also ran experiments on a multigraph extracted from
%co-authorship of published papers in the conferences STOC and FOCS
%from 1964--2001. 
%The results are very similar to those of the MemeTracker data set,
%and therefore omitted here.} 
%%data set. \dkcomment{The footnote may be a candidate for deletion.}
Following the approach of \cite{InfluenceStability}, for both
networks, we assign ``confidence intervals'' 
$\ID{e} = [(1-q) \ActProbD{e}, (1+q) \ActProbD{e}]$,
where the \ActProbD{e} are the inferred activation probabilities. For experiments on the MemeTracker network, we set $q \in \{10\%, 20\%, 30\%, \ldots, 100\%\}$, while we use a coarse grid for the experiments on the large graph STOCFOCS with
$q \in \{5\%, 10\%, 20\%, 50\%, 100\%\}$.

While Lemma~\ref{LEM:INTERVAL-ENDPOINTS} guarantees that the
worst-case instances have activation probabilities 
$(1-q) \ActProbD{e}$ or $(1+q) \ActProbD{e}$, 
this still leaves $2^{|E|}$ candidate functions, too many to include.
We generate an instance for our experiments by sampling 10 of
these functions uniformly, i.e., by independently making each edge's
activation probability either
$(1-q) \ActProbD{e}$ or $(1+q) \ActProbD{e}$.
This collection is augmented by two more instances: 
one where all edge probabilities are $(1-q) \ActProbD{e}$, 
and one where all probabilities are $(1+q) \ActProbD{e}$.
Notice that with the inclusion of these two instances,
the \ALLGREEDY heuristic generalizes the LUGreedy algorithm 
by Chen et al.~\cite{chen:lin:tan:zhao:zhou},
but might provide strictly better solutions on the
  \emph{selected} instances because it explicitly considers those
  additional instances.
The algorithms get to select 20 seed nodes; 
note that in these experiments, we are not considering a bicriteria
approximation. 
\begin{figure}[htb]
\setlength\fboxsep{0pt}
\begin{center}
\subfigure[Different Distributions.]{
\includegraphics[width=0.42\textwidth]{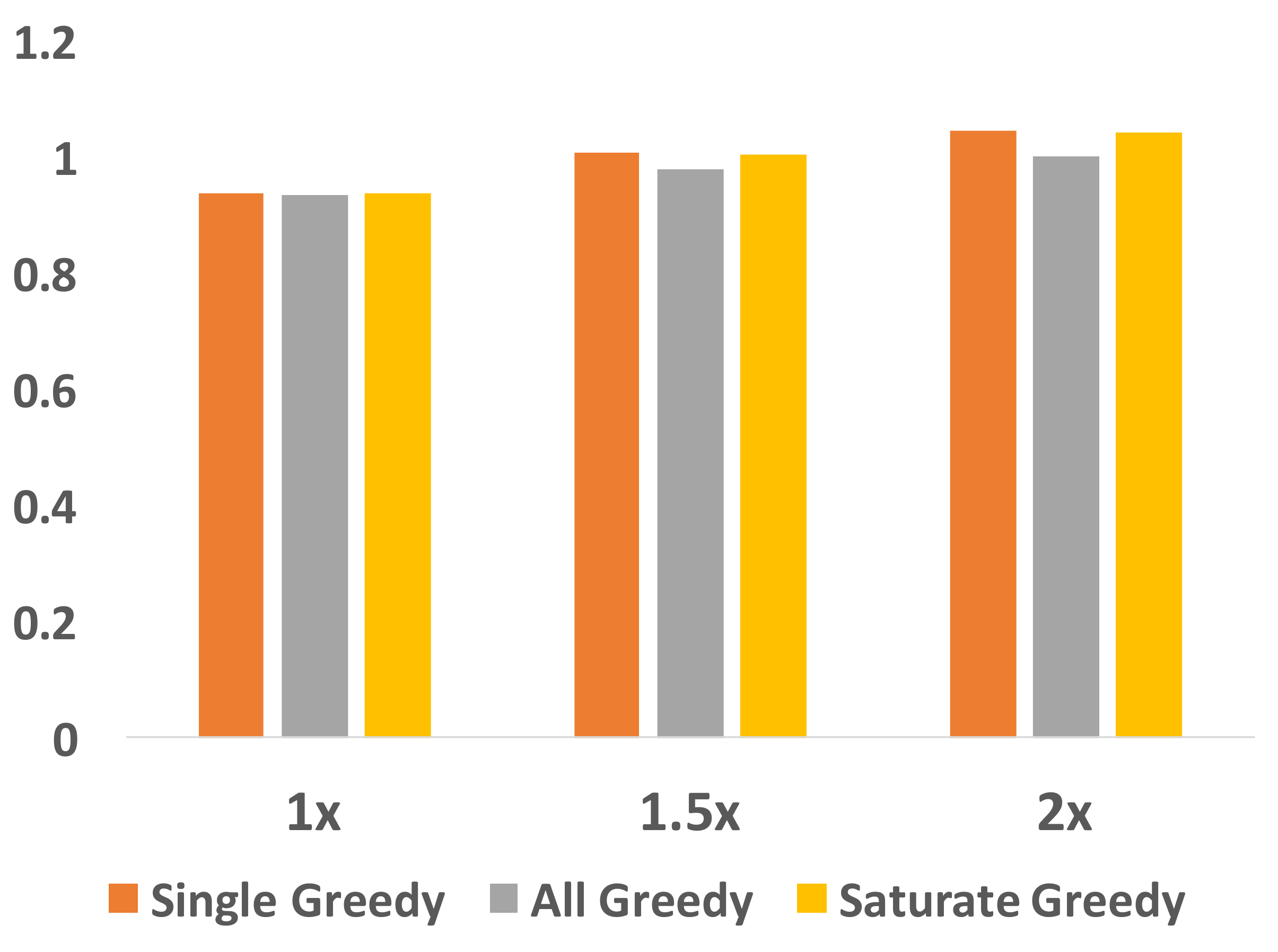}}
\subfigure[Different Models]{
\includegraphics[width=0.42\textwidth]{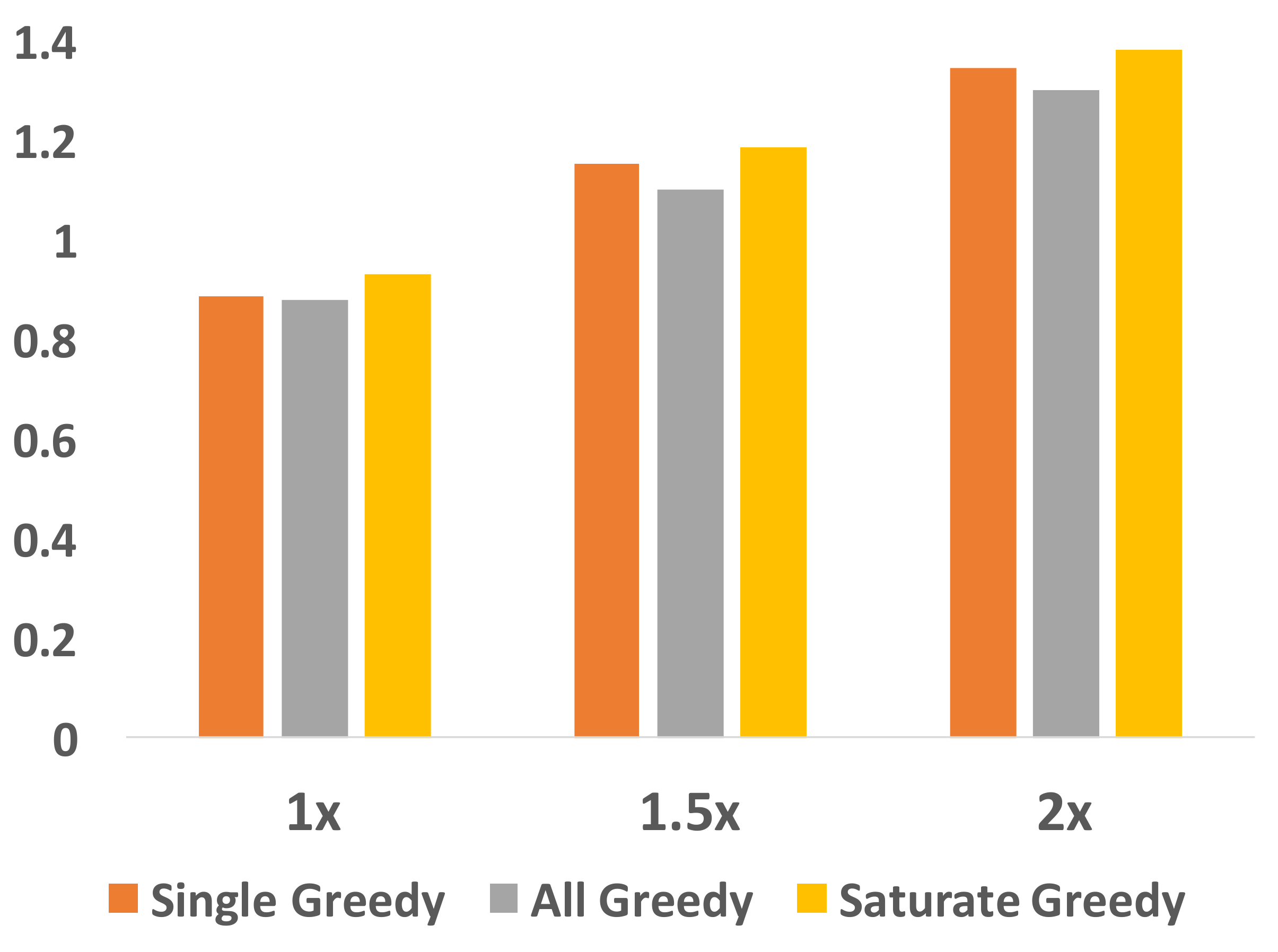}}
\end{center}
\caption{Performance of the algorithms 
(a) under different delay distributions following the CIC model, and 
(b) under different classes of diffusion models. 
The $x$-axis shows the number of seeds selected, and $k=10$.
\label{Fig:diff_models}}
\end{figure}

\begin{figure}[htb]
\begin{center}
%\subfigure[Different Models.]{
%\includegraphics[width=0.32\textwidth]{./figures/diff_model.pdf}}
\subfigure[\IntPer model: MemeTracker ($k=20$)]{
\includegraphics[width=0.45\textwidth]{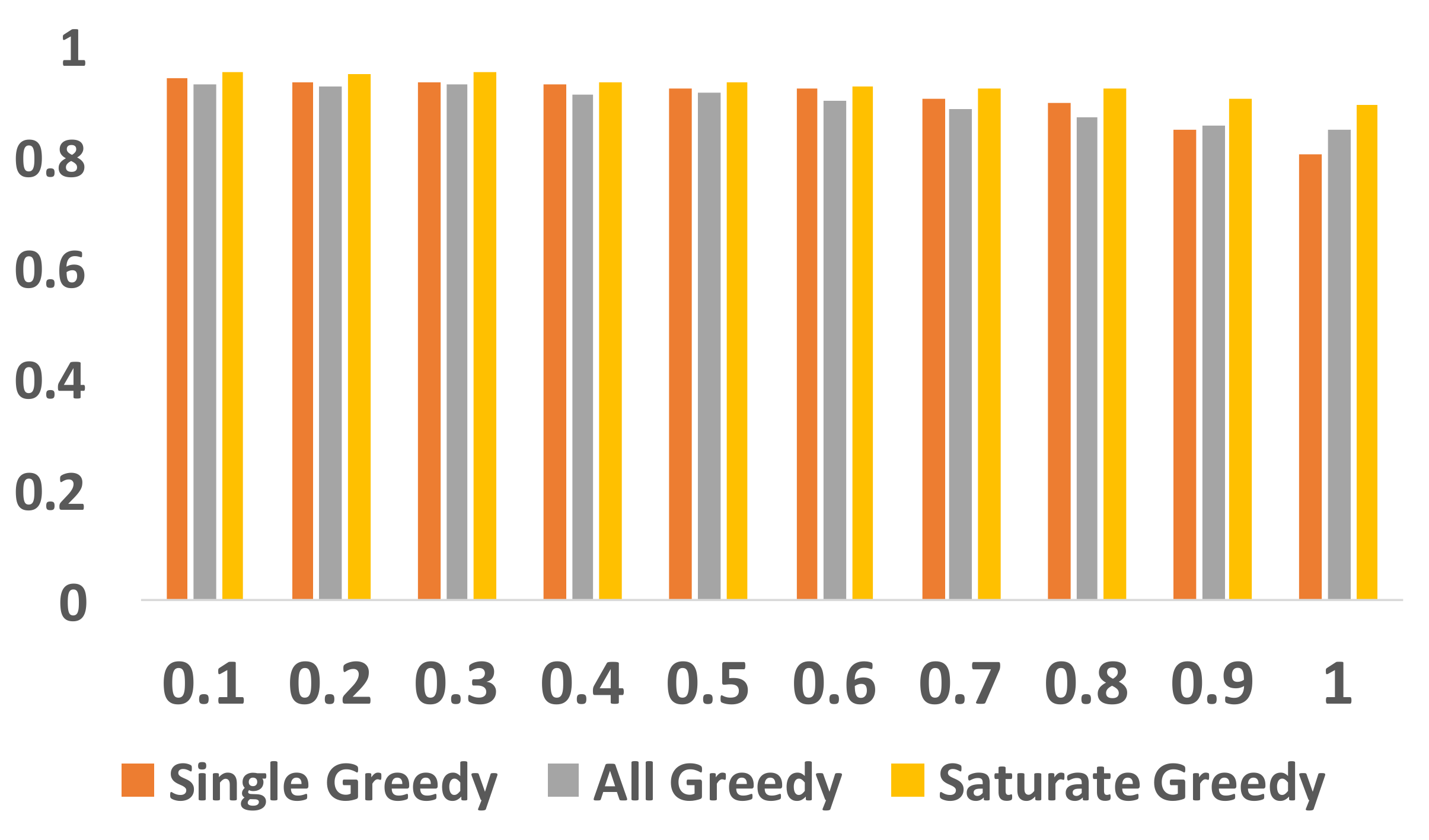}}
\subfigure[\IntPer model: STOCFICS ($k=50$)]{
\includegraphics[width=0.45\textwidth]{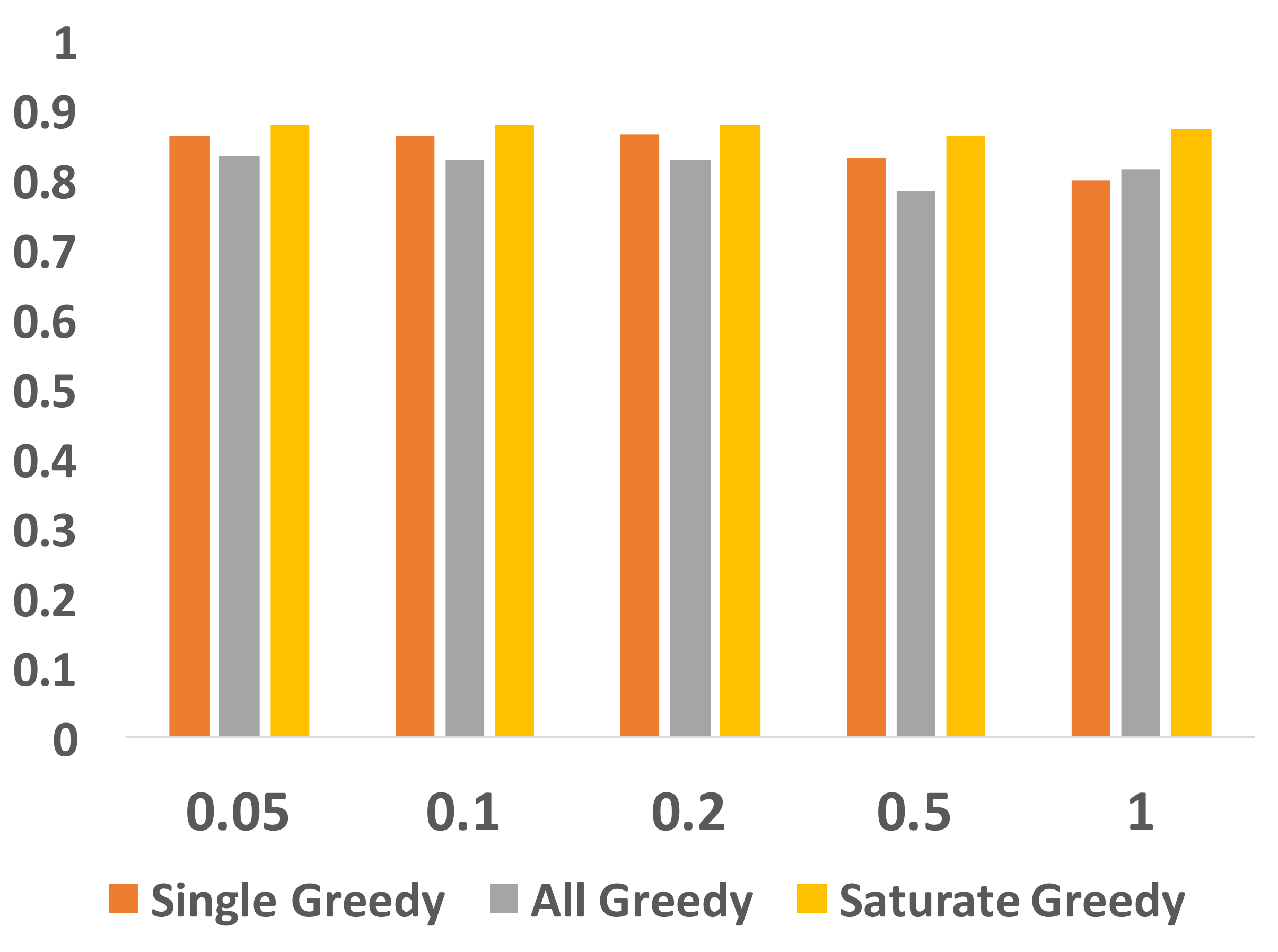}}
\end{center}
\caption{Performance of the algorithms 
%(a) under different diffusion models
%(the $x$-axis shows the number of seeds selected, and $k=10$), and 
under networks sampled from the \IntPer model: (a) MemeTracker networ; (b) STOCFOCS network.
(the $x$ axis shows the (relative) size of the perturbation interval \ID{e}).
\label{Fig:models_perturbation}}
\end{figure}

The results are shown in Figures~\ref{Fig:models_perturbation}(a)
and~\ref{Fig:models_perturbation}(b).
Contrary to the previous results,
when there is a lot of uncertainty about the edge parameters 
(relative interval size 100\% in both networks),
the \SATMINTSS algorithm more clearly outperforms the \GREEDY and
\ALLGREEDY heuristics.
Thus, robust optimization does appear to become necessary when there
is a lot of uncertainty about the model's parameters.

Notice that the evaluation of the algorithms' seed sets is performed
only with respect to the \emph{sampled} influence functions, not with
respect to all $2^{|E|}$ functions. 
Whether one can efficiently identify a worst-case parameter setting
for a given seed set \SeedS is an intriguing open question.
%, discussed briefly in Section~\ref{sec:conclusion}.
Absent this ability, we cannot efficiently guarantee that the solutions are
actually good with respect to all parameter settings.

\subsection{Scalability}\label{sec:scalability}
To evaluate the scalability of the algorithms, we depart from
real-world data sets in order to obtain a controlled environment.
We generate networks using the Kronecker graph model
\cite{Leskovec:Chakrabarti:Kleinberg:Faloutsos:Ghahramani:kronecker}
with either random, core-peripheral or hierarchical-community
structures.
%\footnote{%
%We only present results for networks with random structure. 
%The results for core-peripheral or hierarchical-community networks are
%very similar, and therefore omitted.}
For each type, we generate a set of $5, 10, 15, 20, 25$ networks of
sizes $128, 256, \ldots, 4096$. 
We use the DIC model with activation probability set to $0.1$,
and select $k=50$ nodes.
The running times of the three algorithms are shown in
Figure~\ref{Fig:runningtime} and
Figure~\ref{Fig:runningtime_graphnum}.
In Figure~\ref{Fig:runningtime}, we fix the number of 
networks to five and vary the size of each network;
in Figure~\ref{Fig:runningtime_graphnum}, we fix the size of the networks to
$1024$ and vary the number of networks.
The graphs show that the heuristics are faster than the \SATMINTSS
algorithm by about a factor of ten, but all three algorithms scale
linearly both in the size of the graph and the number of networks, 
due to the fast influence estimation method.
%
%\begin{figure}[htb]
%\setlength\fboxsep{0pt}
%\begin{center}
%\subfigure[Varying \#nodes]{
%\includegraphics[width=0.22\textwidth]{./figures/scalability_randdom.pdf}}
%\subfigure[Varying \#graphs]{
%\includegraphics[width=0.22\textwidth]{./figures/scalability_random_1024_graphs.pdf}}
%\end{center}
%\vskip -0.2in
%\caption{Running times on Kronecker graph networks with random structures. $y$-axis is the running time in seconds, both plotted on a log scale. (a) The $x$ axis represents the number of nodes with number of graphs fixed to five. (b) The $x$ axis represents the number of graphs with number of nodes fixed to $1024$. 
%\label{Fig:runningtime}}
%\end{figure}
\begin{figure*}[htb]
\begin{center}
\subfigure[Random]{
\includegraphics[width=0.32\textwidth]{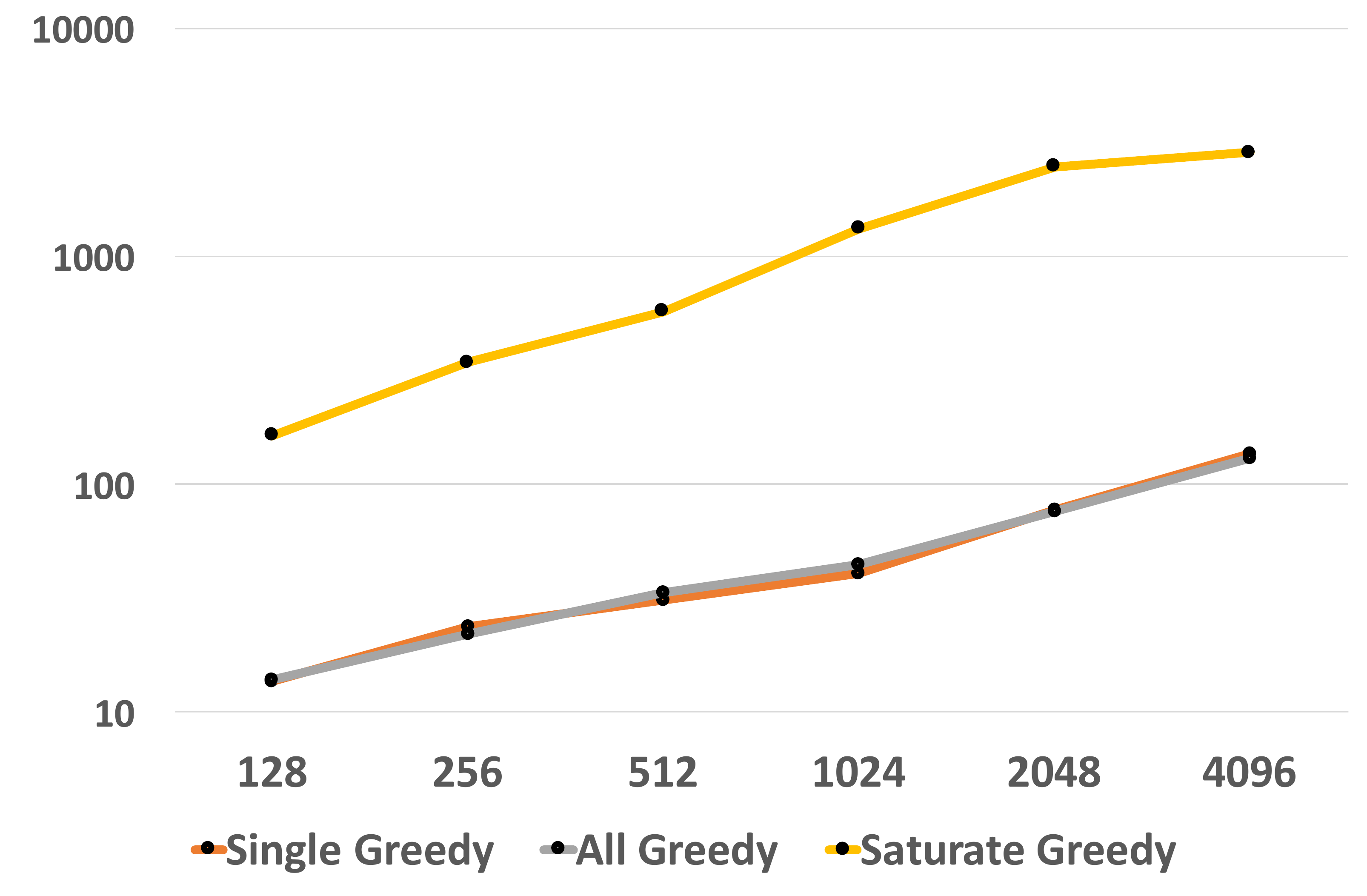}}
\subfigure[Core-peripheral]{
\includegraphics[width=0.32\textwidth]{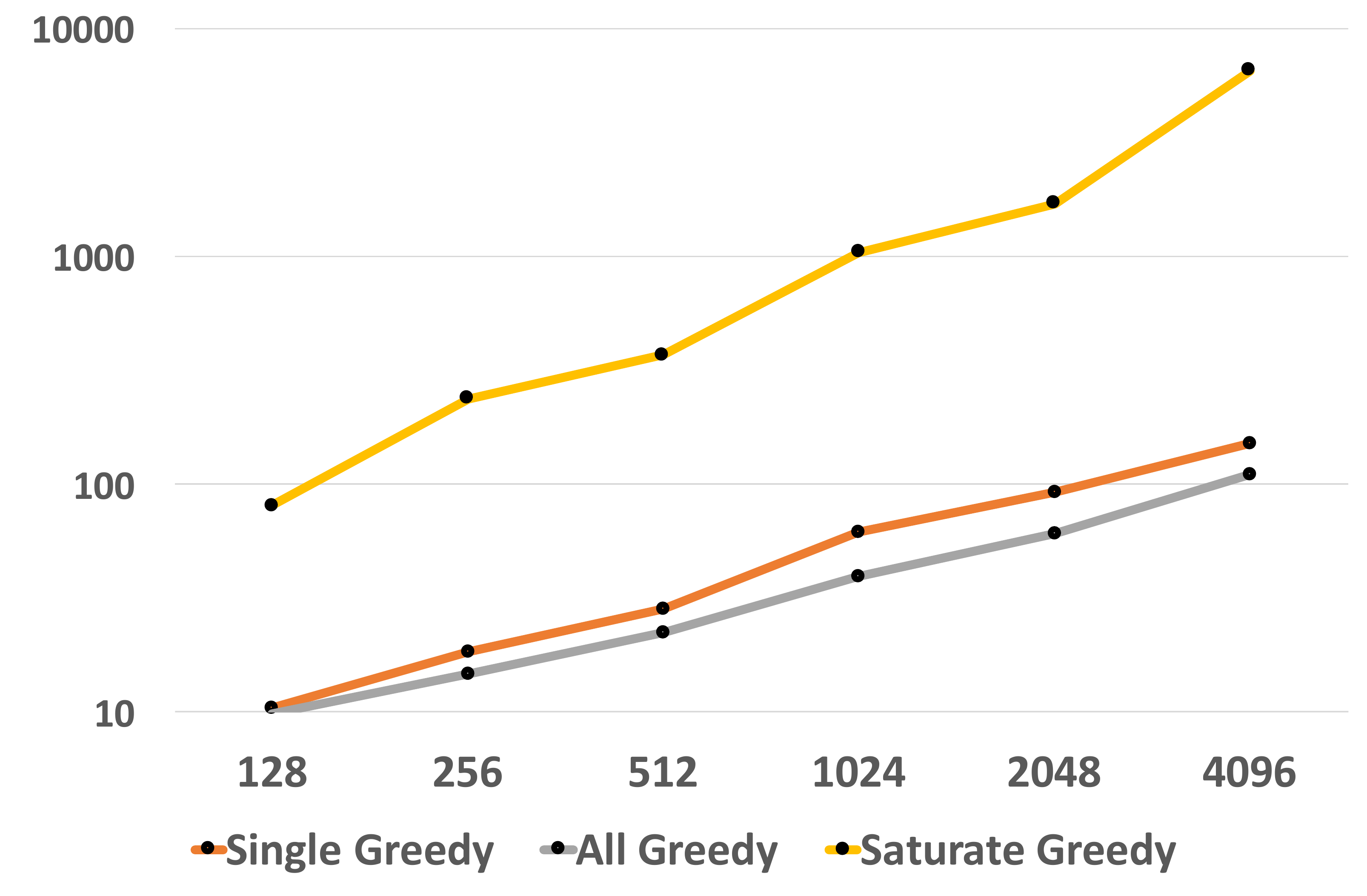}}
\subfigure[Hierarchical-community]{
\includegraphics[width=0.32\textwidth]{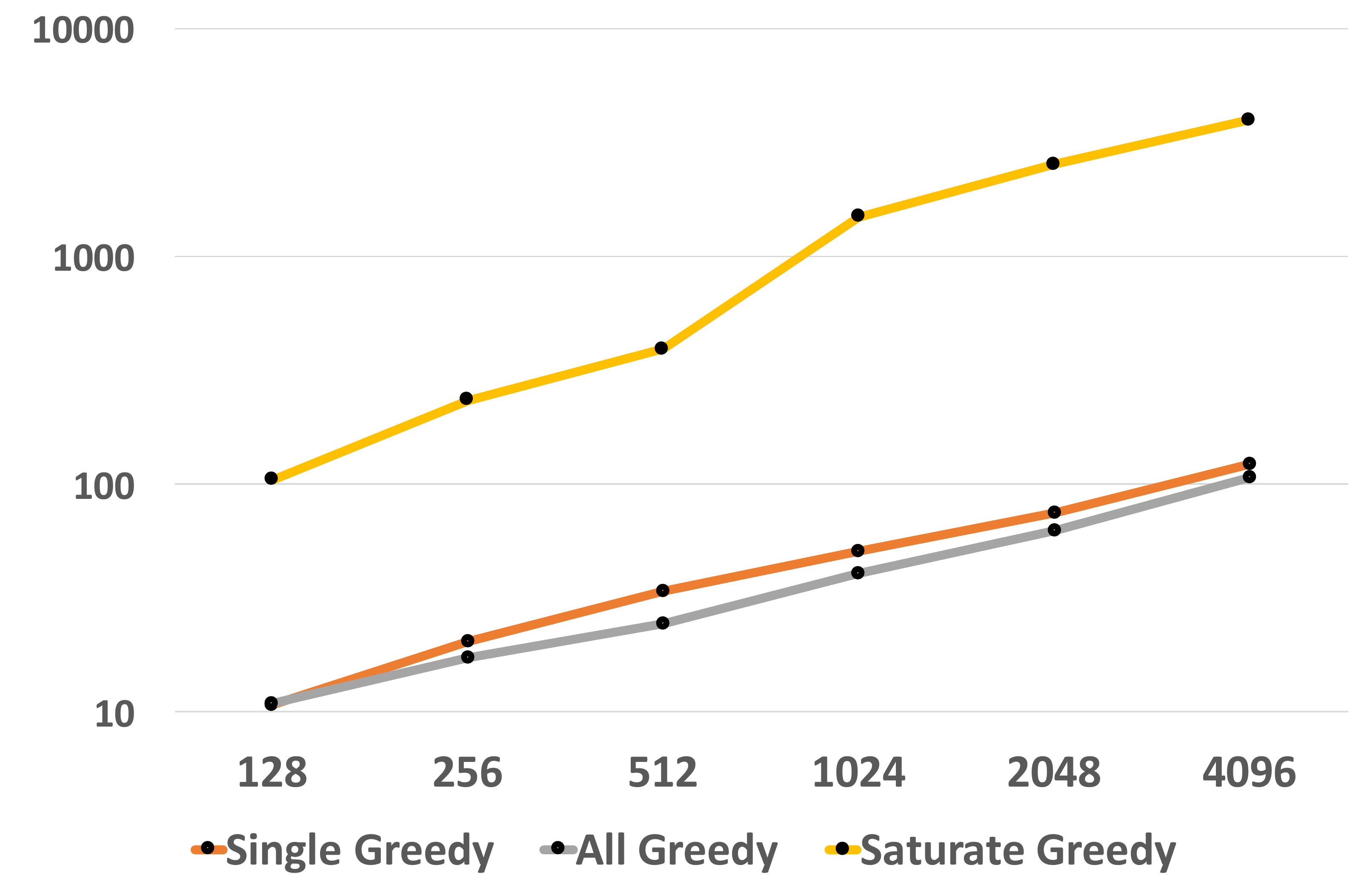}}
\end{center}
\caption{Running times on Kronecker graph networks with different
  structures. The $x$ axis represents the number of nodes, and the
  $y$-axis is the running time in seconds, both plotted on a log
  scale.}\label{Fig:runningtime} 
\end{figure*}
\begin{figure*}[htb]
\begin{center}
\subfigure[Random]{
\includegraphics[width=0.32\textwidth]{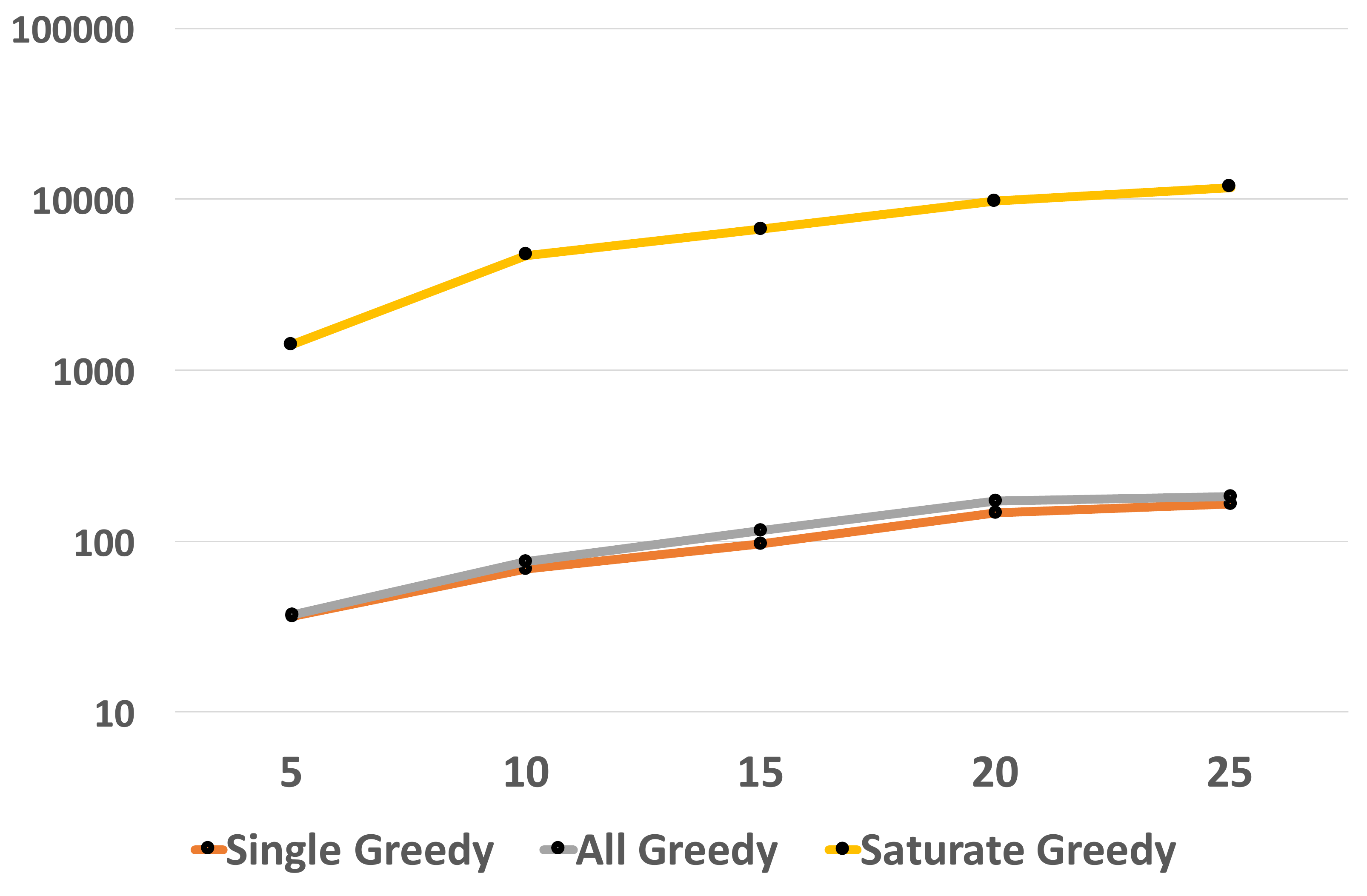}}
\subfigure[Core-peripheral]{
\includegraphics[width=0.32\textwidth]{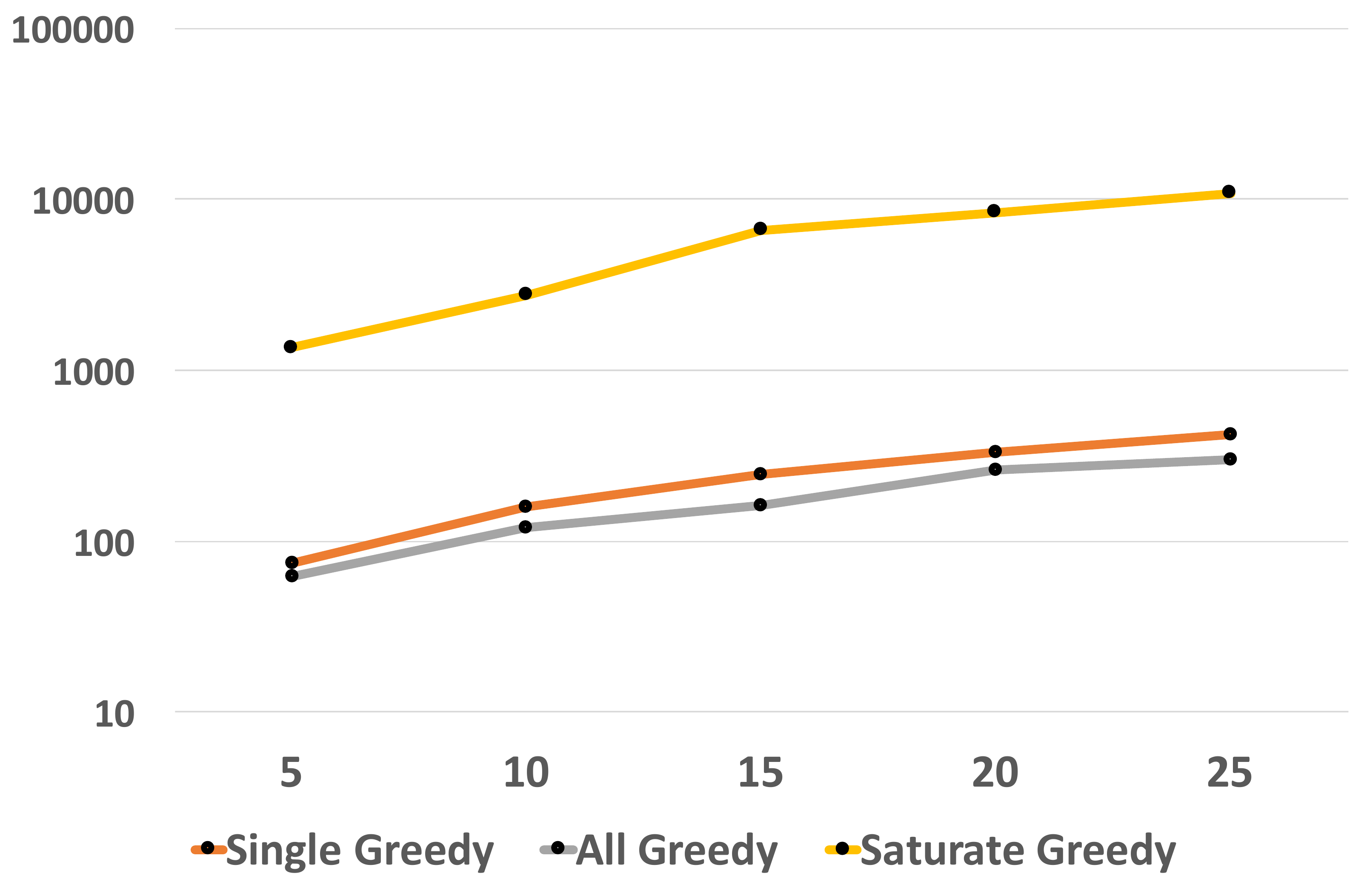}}
\subfigure[Hierarchical-community]{
\includegraphics[width=0.32\textwidth]{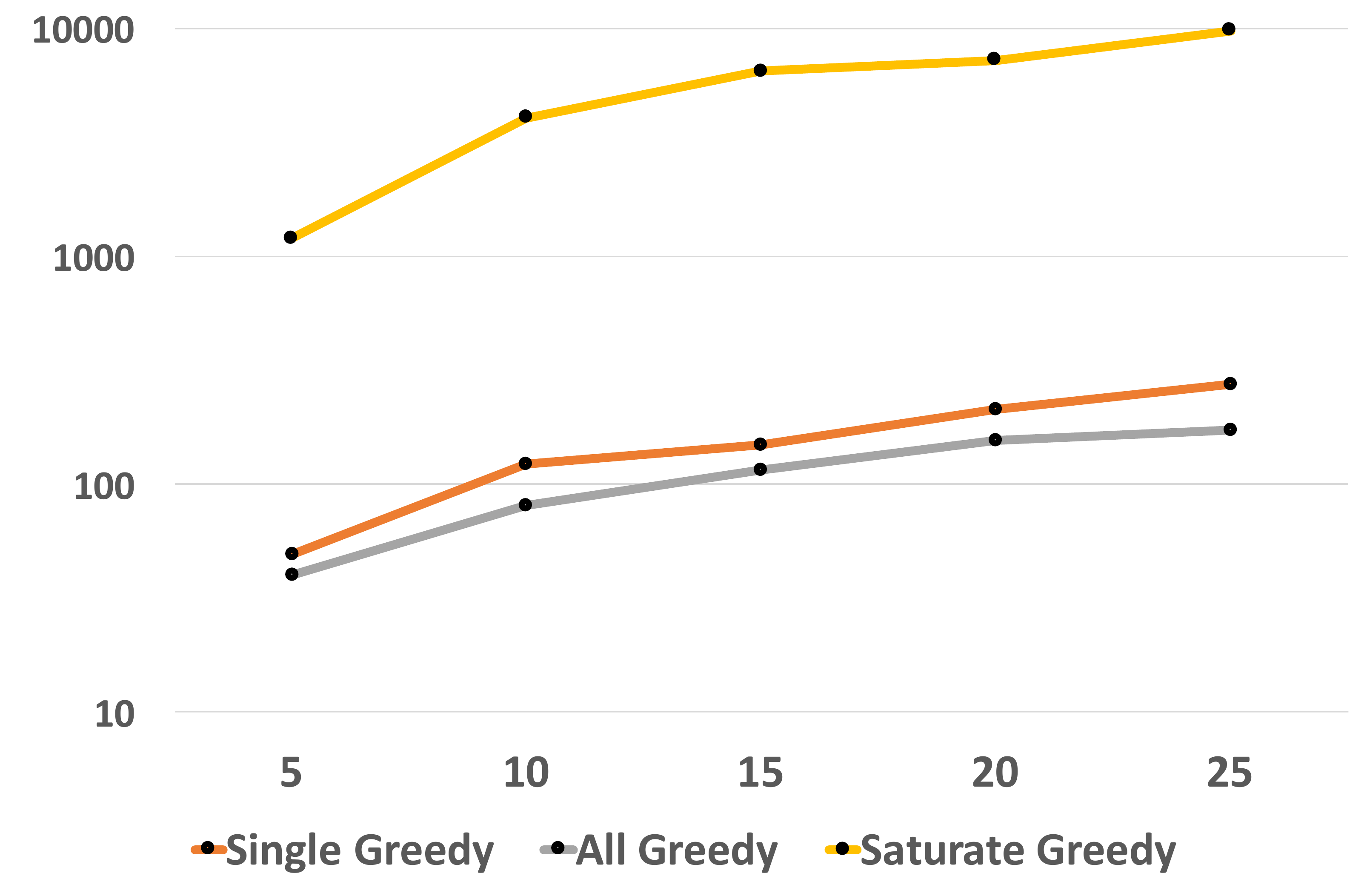}}
\end{center}
\caption{Running times on Kronecker graph networks with different
  structures. The $x$ axis represents the number of diffusion settings, and the
  $y$-axis is the running time in seconds, both plotted on a log
  scale.}\label{Fig:runningtime_graphnum} 
\end{figure*}

\section{Future Work} \label{sec:conclusion}
Our work marks an early step, rather than the conclusion, in devising 
robust algorithms for social network tasks, and more specifically \InfMax.
An interesting unresolved question is whether one can efficiently find
an (approximately) worst-case influence function in the \IntPer model.
This would allow us to empirically evaluate the performance of natural
heuristics for the \IntPer model, such as randomly sampling a small
number of influence functions.
Furthermore, it would allow us to design ``column generation'' style
algorithms for the \IntPer model, where we alternate between finding
a near-optimal seed set for all influence functions encountered so
far, and finding a worst-case influence function for the current seed
set, which will then be added to the encountered functions.
%While our hardness results (for both problems) rule out certain types
%of (bicriteria) approximation tradeoffs, other tradeoffs are
%conceivable, e.g., exceeding the target seed set size by a factor
%$\omega(\log n)$.

In the context of the bigger agenda, one could conceive of other
notions of robustness in \InfMax, perhaps tracing a finer line between
worst-case and Bayesian models.
Also, much more research is needed into identifying which influence
models best capture the behavior of real-world cascades, and under
what circumstances. It is quite likely that different models will
perform differently depending on the type of cascade and many other
factors, and in-depth evaluations of the models could give
practitioners more guidance on which mathematical models to choose.
While our model of robustness allows us to combine instances
  of different models (e.g., IC and LT), this may come at a cost of
  decreased performance for each of the models individually. 
Thus, it remains an important task to identify the influence models
that best fit real-world data.

\subsubsection*{Acknowledgments}
We would like to thank Shaddin Dughmi for useful pointers and
feedback, and Shishir Bharathi and Mahyar Salek for useful
discussions, and anonymous reviewers for useful feedback. 
The research was sponsored in part
by NSF research grant IIS-1254206
and by the U.S.~Defense Advanced Research Projects Agency 
(DARPA) under Social Media in Strategic Communication
(SMISC) program, Agreement Number W911NF-12-1-0034. 
The views and conclusions are those of the authors and should not be
interpreted as representing the official policies of the funding agency,
or the U.S. Government.

\bibliographystyle{plain}
%\bibliography{bibliography/names,bibliography/conferences,bibliography/publications,bibliography/bibliography,bibliography/additional}
\bibliography{robust_infmax.bbl}

\appendix

\section{Proof of Theorem~\ref{THM:HARDNESS}}
\label{sec:hardness}
\newcommand{\ALLSETS}{\ensuremath{\mathcal{T}}\xspace}

%\begin{extraproof}{Theorem~\ref{thm:hardness}}
We prove the two parts of the theorem
by (slightly different) reductions from the gap version of \SETCOVER.
A \SETCOVER instance consists of a universe
$U=\Set{a_1, \ldots, a_N}$, 
a collection $\mathcal{T}$ of $M$ subsets of $U$,
and an integer $k$.
A set cover is a collection $\mathcal{C} \subseteq \ALLSETS$ such
that $\bigcup_{T \in \mathcal{C}} T = U$.
Without loss of generality, we assume that each element is contained
in at least one set --- otherwise, there trivially is no set cover.
Also, without loss of generality, we assume that $k \leq \min(M,N)$, 
as otherwise, one can trivially pick all sets or one designated set
per element.

The gap version of \SETCOVER then asks us to decide whether
there is a set cover $\mathcal{C}$ of size 
$|\mathcal{C}| \leq k$ or whether each set cover has size at least
$(1-\delta) \ln N \cdot k$.
(The algorithm is promised that the minimum size will never lie
between these two values.)
Dinur and Steurer~\cite[Corollary 1.5]{Dinur:Steurer:setcover} showed
that the gap version of \SETCOVER is NP-hard.

\paragraph{Part 1}
Based on the \SETCOVER instance, we construct the following
instance of \RobInfMax under the DIC model.
Let $m := (\max(N,M))^{3/\epsilon}$.
The instance consists of $N$ bipartite graphs on a shared vertex set
$V = X \cup Y$. 
$X$ contains one node $x_T$ for each set $T \in \ALLSETS$;
$Y$ contains $m$ nodes $y_{a,1}, \ldots, y_{a,m}$ 
for each element $a \in U$.
Hence, the number of nodes in the constructed graph is 
$n = M+mN = \Theta(mN) \leq \Theta(m^{1+\epsilon/3})$; 
in particular, it is polynomial, and the reduction takes polynomial time.

In the \Kth{i} influence function, all nodes 
$x_T$ with $T \ni a_i$ have a directed edge with activation
probability 1 (or exponential delay distribution with delay parameter 1) to all of the $y_{a_i,j}$ (for all $j$); 
no other edges are present.
Hence, $|\IMFSET| = N$, and $\ln N = \ln |\IMFSET|$.
For the CIC model, the time window has size $T=NM$.

First, consider the case when there is a set cover $\mathcal{C}$ of
size $k$. 
Choose the corresponding $x_T, T \in \mathcal{C}$ as seed nodes,
and call the resulting seed set $S$. 
Because $\mathcal{C}$ is a set cover, in the \Kth{i} instance, 
%because $\mathcal{C}$ is a set cover, 
%there is at least one $T \ni a_i$ with $T \in \mathcal{C}$.
%Therefore, in the \Kth{i} instance, 
all of the $y_{a_i,j}$ are activated, for a total of at least $m+k$ nodes.
(Under the CIC model, all of these $y_{a_i,j}$ are activated
  with high probability, not deterministically, within the $T$ steps)
Because none of the nodes in $X$ and none of the 
$y_{a_{i'},j}, i' \neq i$ have incoming edges in the \Kth{i} instance,
the optimum solution for that instance can activate at most all of the
$m$ nodes $y_{a_i, j}$ and its $k$ selected nodes, for a total of $m+k$.
Thus, the objective function value will be 1 (or arbitrarily
  close to 1 w.h.p.~for the CIC model).

Now assume that there is no set cover of size 
$(1-\delta) \ln N \cdot k$, and consider any seed set $S$.
Let $k' = |S \cap X| \leq (1-\delta) \ln N \cdot k$ 
be the number of nodes from $X$ selected as seeds.
Because the set $\mathcal{S} := \Set[x_T \in S]{T \in \ALLSETS}$
cannot be a set cover by assumption, there must be some
$a_i \notin \bigcup_{T \in \mathcal{S}} T$.
Therefore, under the the \Kth{i} influence function, 
none of the $y_{a_i,j}$ can be ever activated, except those selected
directly in $S$. 
Hence, the number of nodes activated under the \Kth{i} influence
function is at most $|S| \leq (1-\delta) \ln N \cdot k$.
On the other hand, by selecting just one node $x_T$ corresponding to
any set $T \ni a_i$, one could have activated all of the
$y_{a_i,j}$ (with high probability under the CIC model), for a total of $m$.
Thus, the objective function value is at most
$\WCObj{S} \leq \frac{(1-\delta) \ln N \cdot k}{m} \leq O(m^{2\epsilon/3-1})
\leq O(n^{\frac{2\epsilon-3}{3+\epsilon}}) = o(n^{-(1-\epsilon)})$,
where we crudely bounded both $\ln N$ and $k$ by $N \leq O(m^{\epsilon/3})$.
%
%In summary, if the \SETCOVER gap instance has a small set cover, then
%there is a seed set $A$ with $\WCObj{A} \geq 1$, 
%while otherwise, for every seed set $A$, we have
%$\WCObj{A} \leq m^{2\epsilon/3-1}$.
%We re-express this upper bound on \WCObj{A} in terms of the number $n$ of
%nodes as $n^{\frac{2\epsilon-3}{3+\epsilon}} = o(n^{-(1-\epsilon)})$.

Hence, a $((1-\delta) \ln N, O(n^{1-\epsilon}))$ bicriteria
approximation algorithm could distinguish the two cases, 
and thus solve the gap version of \SETCOVER. 

\paragraph{Part 2}
For the second part, we just consider the gap version with a fixed
$\delta$, say, $\delta = \frac{1}{2}$.
Then, in the hard instances, $M$ and $N$ are polynomially related,
which we assume here, i.e., $M \leq N^q$ for some constant $q$ which
is independent of $\epsilon$ or $N$.

Based on the \SETCOVER instance, 
we construct a different \RobInfMax instance, 
consisting of a directed graph with three layers $V = X \cup Y \cup Z$. 
The first layer again contains one node $x_T$ for each set 
$T \in \ALLSETS$;
the second layer now contains just one node $y_a$ for each element $a \in U$.
There is an edge (with known influence probability 1,
or exponential delay distribution with parameter 1) 
from $x_T$ to $y_a$ if and only if $a \in T$.
The third layer $Z$ contains $m = (\max(N,M))^{2/\epsilon}$ nodes.
For each $a \in U$ and $z \in Z$, there is a directed edge
$(y_a, z)$ with complete uncertainty about its parameter: 
under the DIC model, the probability is in the interval
 $\ID{(y_a,z)} = [0,1]$, and under the CIC model, the edge delay is exponentially distributed
with parameter in the interval $\ID{(y_a,z)} = (0,1]$.
In total, the graph has $n = N+M+m = \Theta(m)$ nodes 
(in particular, polynomially many),  
and the reduction takes polynomial time.
Because $N$ is at most polynomially smaller than $M$, 
we have $N = \Omega(n^{\epsilon/2q})$,
and thus $\ln N = \Omega(\frac{\epsilon}{q} \cdot \ln(n))$.
For the CIC model, we set the time horizon to $T = NM$.

First, consider the case when there is a set cover $\mathcal{C}$ of
size $k$.
Consider choosing the corresponding $x_T, T \in \mathcal{C}$ as seed
nodes; call the resulting seed set $S$. 
$S$ will definitely activate all nodes in $Y$, for a total of $k+N$.
Now, consider any assignment of probabilities \ActProbD{y_a,z} or edge delays \Delay{y_a,z} to
the edges from $Y$ to $Z$, and an optimal seed set $S^*$ of size $k$.
Let $Z^* = Z \cap S^*$ be the set of seed nodes chosen from $Z$, of size
$k'$. Then, $S^*$ definitely activates all of $Z^*$, and at most
all $N$ nodes from $Y$ as well as $k-k'$ nodes from $X$, for a total
(so far) of $N+k$.
For any node $z \in Z \setminus Z^*$, the probability that it is
activated by $S$ is at least as large as under $S^*$, because
for any values of the individual activation probabilities or delays
between $Y$ and $Z$, the fact that $S$ activates all of $Y$ ensures
that any node in $Z$ activated under $S^*$ is also activated under $S$
(by time $T$, in the case of the CIC model).
%activated by $S$ is $1-\prod_{a} (1-\ActProbD{y_a,z})$.
%On the other hand, under the solution $S^*$, at most all of $Y$ is
%activated, giving a probability of at most
%$1-\prod_{a} (1-\ActProbD{y_a,z})$ for activation of $z$.
Because, the expected number of nodes activated from $Z \setminus Z^*$ is
at least as large under $S$ as under $S^*$, and the ratio is $1$.
Since this holds for all settings of the activation
  probabilities or edge delay parameters,
we get that $\WCObj{S} \geq 1$.

Now assume that there is no set cover of size 
$\frac{1}{2} \ln N \cdot k$, and consider any seed set $S$.
If $S$ contained any node $y_a$, we could replace it with any node
$x_T$ such that $a \in T$ and activate at least as many nodes as
before, so assume without loss of generality that 
$S \cap Y = \emptyset$. 
Because $|S \cap X| \leq |S| \leq \frac{1}{2} \ln N \cdot k$,
the gap guarantee implies that there is at least one node
$y_a \in Y$ that is never activated by $S$.
%(because $a \notin T$ for any set $T$ corresponding to nodes in $A$).
Now consider the probability assignment $\ActProbD{y_a,z} = 1$ for all
$z \in Z$, and $\ActProbD{y,z} = 0$ for all $y \neq y_a, z \in Z$.
(Under the CIC model, set $\Delay{y_a,z} = 1$ for all $z \in Z$, 
and $\Delay{y,z} = 1/(NM)^2$ for all $y \neq y_a, z \in Z$.)
Then, the seed set $S$ cannot activate any nodes in $Z$ 
(except those it may have selected), and will activate a total of at most
$N + k = O(N) = O(n^{\epsilon/2})$ nodes.
(Under the CIC model, this statement holds with high probability.)
On the other hand, the seed set $\Set{y_a}$ (just a single node) would
have activated all of $Z$ (with high probability, under the
  CIC model), for a total of $m+1 = \Omega(n)$ nodes. 
Hence, the ratio is at most $O(n^{\epsilon/2}/n) = O(1/n^{1-\epsilon/2})$,
implying that $\WCObj{S} \leq O(1/n^{1-\epsilon/2})$.

If there were an $(\epsilon \cdot c \cdot \ln(n), O(n^{1-\epsilon}))$ bicriteria
approximation algorithm for a sufficiently small constant $c$, 
it could distinguish which of the two cases 
($\WCObj{S} = 1, \WCObj{S} \leq O(1/n^{1-\epsilon/2})$) applied,
thus solving the gap version of \SETCOVER.

\end{document}